\documentclass[a4paper,11pt]{article}

\usepackage{geometry}
\usepackage[bitstream-charter]{mathdesign}
\usepackage[utf8]{inputenc}
\usepackage[T1]{fontenc}
\usepackage[english]{babel}
\usepackage{amsthm}
\newtheorem{theorem}{Theorem}[section]

\usepackage{algorithm}
\usepackage{algpseudocode}

\usepackage{graphicx}
\usepackage{float}
\usepackage{tikz}
\usetikzlibrary{graphs}
\usetikzlibrary{calc}

\usepackage{amsmath}

\usepackage{hyperref}

\usepackage[export]{adjustbox}

\newcommand\variablename[1]{\mathop{\mathit{#1}}\nolimits}
\newcommand{\opt}{\variablename{opt}}
\newcommand{\onl}{\variablename{onl}}

\begin{document}

\pagenumbering{gobble}
\begin{titlepage}
  \begin{center}
    {\huge{Multi-Agent Online Graph Exploration on Cycles and Tadpole Graphs}\footnote{A summary of this work was published at SIROCCO 2024 as a Brief Announcement~\cite{BriefAnnouncement}}}
             
    \vspace{1cm}

		\centering
		Erik van den Akker, Kevin Buchin, Klaus-Tycho Foerster 
				
		\vspace{0.3cm}
		TU Dortmund, Germany

    \vspace{1cm}
    \begin{abstract}
      We study the problem of multi-agent online graph exploration, in which a team of $k$ agents has to explore
      a given graph, starting and ending on the same node. The graph is initially unknown. Whenever a node is visited by an agent, its neighborhood and adjacent edges are revealed. 
      The agents share a global view of the explored parts of the graph. The cost of the exploration
      has to be minimized, where cost either describes the time needed for the entire exploration (time model), or the length of the longest path traversed by any
      agent (energy model).
      We investigate graph exploration on cycles and tadpole graphs for 2-4 agents, providing optimal
      results on the competitive ratio in the energy model ($1$-competitive with two agents on cycles and three agents on tadpole graphs),
      and for tadpole graphs in the time model ($1.5$-competitive with four agents). 
      We also show competitive upper bounds of $2$ for the exploration of tadpole graphs with three agents, 
      and $2.5$ for the exploration of tadpole graphs with two agents in the time model. 
    \end{abstract}
  \end{center}
\end{titlepage}
  
\pagenumbering{arabic}
\setcounter{page}{1}
\section{Introduction} \label{introduction}

In the \textit{Online Graph Exploration} problem, an agent is initially
placed onto a random node of an unknown but labeled weighted graph. The agents task is to visit
all nodes of the graph and return to the starting node.
Since the agent gets the information about the graph online, it only knows about the already visited nodes,
as well as the neighborhood of these nodes at any point in time. Whenever a new node is visited,
the adjacent edges and the neighborhood of the visited node are revealed to the agent.
Traversing an edge incurs the cost of its weight. The goal of the agent is to explore the graph
as efficiently as possible, trying to keep down the cost of the exploration.
This problem was initially described as \textit{the online traveling salesperson problem} 
by Kalyanasundaram and Pruhs \cite{Kalyanasundaram1993}.\\
We examine a modified version of this problem using multiple agents starting at the same node, which
have to collaboratively explore a graph so that each node is visited at least once by some agent.
This variation does not have a consistent name in the literature and will here be referred to as 
\textit{Multi-Agent Online Graph Exploration}.

\subsection{Motivation}

The Online Graph Exploration problem has been widely used to describe the exploration of unknown terrain by multiple autonomous robots like mars rovers \cite{Dynia}, or robot vacuums \cite{Ortolf2012},
where the explored terrain is modeled as an unknown graph. It is also used to describe fleets of robots fully exploring unknown caves or damaged structures, to find and rescue human beings \cite{Higashikawa2012}. 
The literature contains multiple variations of this problem in terms of cost and communication models. Some allow agents to freely
communicate, while in some variations the agents have to meet to exchange information, or exchange information by writing and reading to/from the nodes they are currently located at. Some variations focus on the general exploration time (time model), while others 
focus on the maximum distance traversed by an individual agent. This model is used to describe the maximum energy consumption of any single agent during an exploration (energy model).\\
In the single-agent case, many graph classes like directed graphs \cite{Foerster2016}, tadpole graphs \cite{Brandt2020} or cactus graphs \cite{Fritsch2021} have already been investigated, 
while the main focus of the literature in the case of multi-agent exploration lies on general graphs and trees, with the exception of cycles in the time model \cite{Higashikawa2012} and $n\times n$-grid graphs \cite{Ortolf2012}.\\
The goal of this paper is to analyze exploration strategies for other restricted graph classes in terms of their competitive ratios, which describe the relationship between the exploration time
achieved by an online algorithm and the optimal result of an (offline) algorithm with full knowledge of a given graph. The analysis will be done for both, the time and the energy model. 
The particular graph classes considered are cycles (in the energy model) and tadpole graphs, which consist of a cycle and a path attached to one of its nodes.

\subsection{Contribution}

We show that the exploration strategy of avoiding the longest edge has a competitive ratio of $1.5$ on cycles in the energy model with two agents and a competitive ratio of at least $2$
on tadpole graphs in the time model, when using three agents. We then provide a modified exploration strategy, which does not decide based on the next edge weights, but on the length of the paths traversed
by the agents up to that point. On cycles this modified strategy still has a competitive ratio of $1.5$ in the time model, but also achieves a competitive ratio of $1$ in the energy model using two agents.
On tadpole graphs our exploration strategy achieves a competitive ratio of $1.5$ with a team of four agents in the time model, 
and a competitive ratio of $1$ with a team of three agents in the energy model. With a team of two agents, we achieve a competitive ratio
of $2.5$ on tadpole graphs in both models. An overview of the results is given in Table \ref{resultstable}.

\begin{table}[H]
    \small
    \begin{adjustbox}{center}
    \begin{tabular}{|l|ll|ll|}
    \hline
                                  & \multicolumn{2}{l|}{Time Model}                                     & \multicolumn{2}{l|}{Energy Model}              \\ \hline
    Graph Class (\# of Agents)                            & \multicolumn{1}{c|}{Lower Bound} & \multicolumn{1}{c|}{Upper Bound} & \multicolumn{1}{l|}{Lower Bound} & Upper Bound \\ \hline
    Cycles ($2$ Agents)             & \multicolumn{1}{l|}{$1.5$ \cite{Higashikawa2012}}         & $1.5$ \cite{Higashikawa2012}                              & \multicolumn{1}{l|}{$1$}           & $1$ (Thm. \ref{ampCycles})           \\ \hline
    Tadpole Graphs ($2$ Agents)     & \multicolumn{1}{l|}{$1.5$ (Thm. \ref{TadpoleTimeLower})}        & $2.5$ (Thm. \ref{TadpoleUpperTwo})                             & \multicolumn{1}{l|}{$1.5$ (Thm. \ref{TadpoleEnergyLower})}         & $2.5$ (Thm. \ref{TadpoleUpperTwo})      \\ \hline
    Tadpole Graphs ($3$ Agents)     & \multicolumn{1}{l|}{$1.5$ (Thm. \ref{TadpoleTimeLower})}          & $2$ (Thm. \ref{TadpoleUpperThree})                             & \multicolumn{1}{l|}{$1$}           & $1$ (Thm. \ref{TadpoleUpperThree})          \\ \hline
    Tadpole Graphs ($4+$ Agents)     & \multicolumn{1}{l|}{$1.5$ (Thm. \ref{TadpoleTimeLower})}       & $1.5$ (Thm. \ref{TadpoleUpperFour})                             & \multicolumn{1}{l|}{$1$}           & $1$ (Thm. \ref{TadpoleUpperThree})          \\ \hline
    \end{tabular}
    \end{adjustbox}
    \caption{The new lower and upper bounds for the exploration of cycles and tadpole graphs.}
    \label{resultstable}
\end{table}

\subsection{Related Work}

\paragraph*{Single-Agent Graph Exploration} Online Graph Exploration with a single agent has been studied extensively for many different graph classes: For planar
graphs Kalyanasundaram and Pruhs developed the $16$-competitive \textit{ShortCut} algorithm \cite{Kalyanasundaram1993},
which was generalized by Megow et al.\ to the \textit{Blocking} algorithm, which has a competitive ratio of $16(1 + 2g)$ for
any graph with genus $g$ \cite{Megow2012}.\\
Miyazaki et al.\ gave optimal algorithms for cycles and unit-weight graphs, with competitive ratios and lower bounds of
$1+\sqrt{3}$ and $2$ respectively. Megow et al.\ have extended the result on unit-weight graphs, showing that online exploration of graphs with $c$ 
distinct weights is $2c$ competitive \cite{Megow2012}. Brandt et al.\ have shown that greedy exploration is $2$-competitive
for the case of tadpole graphs \cite{Brandt2020}, Fritsch examined the \textit{Blocking} algorithm on unicyclic and cactus graphs
and proved that it is $3$-competitive for unicyclic and $\frac{5}{2}+\sqrt{2}$ competitive on cactus graphs \cite{Fritsch2021}. 
Foerster and Wattenhofer analyzed different variations of directed graphs and provided matching upper and lower bounds of $n-1$ for 
deterministic algorithms on general directed graphs, where the lower bound could be improved to $\frac{n}{4}$ when using randomized
algorithms \cite{Foerster2016}. Very recently Baligács et al.\ have shown that the \textit{Blocking} algorithm achieves constant competitive
ratio on minor-free graphs \cite{Baligacs2023}.\\
On general graphs the best known strategy is the \textit{Nearest Neighbor} algorithm, with a competitive ratio of $\Theta(\log(n))$ shown by Rosenkrantz et al.\ \cite{Rosenkrantz1977}.
The best known lower bound of $\frac{10}{3}$ was shown by Birx et al.\ and holds for general graphs and planar graphs as well. \cite{Birx2020}. An overview of the graph classes that have
already been investigated for single-agent graph exploration is provided in Table \ref{single-agent-table}.

\begin{table}[H]
    \small
    \centering
    \begin{tabular}{|l|l|l|}
    \hline
                       & \multicolumn{1}{c|}{Lower Bound}                                                  & \multicolumn{1}{c|}{Upper Bound} \\ \hline
    General Graphs     & $10/3$ \cite{Birx2020}                                                                               & $O(\log(n))$ \cite{Rosenkrantz1977}                           \\ \hline
    Planar Graphs     & $10/3$ \cite{Birx2020}                                                                              & $16$ \cite{Kalyanasundaram1993}                              \\ \hline
    Unit-weight Graphs  & $2$\cite{MIYAZAKI2009}                                                                                 & $2$\cite{MIYAZAKI2009}                                \\ \hline
    $c$ distinct weights  & $2$\cite{MIYAZAKI2009}                                                                                 & $2c$\cite{Megow2012}                                \\ \hline
    Cycles             &  $\frac{1 + \sqrt{3}}{2}$\cite{MIYAZAKI2009}                                                                           & $\frac{1 + \sqrt{3}}{2}$\cite{MIYAZAKI2009}                           \\ \hline
    Tadpole Graphs    & $2$ \cite{Brandt2020}                                                                                  & $2$ \cite{Brandt2020}                                 \\ \hline
    Unicyclic Graphs   & $2$ \cite{Brandt2020}                                                                                  & $3$ \cite{Fritsch2021}                                \\ \hline
    Cactus Graphs  & $2$ \cite{Brandt2020}                                                                                  & $2.5+\sqrt{2}$ \cite{Fritsch2021}                           \\ \hline
    Directed Graphs    & \begin{tabular}[c]{@{}l@{}}$n-1$ (deterministic)\\$n/4$ (randomized)\end{tabular} \cite{Foerster2016} & $n-1$\cite{Foerster2016}                              \\ \hline
    \end{tabular}
    \caption{Current results for single-agent graph exploration.}
    \label{single-agent-table}
    \end{table}

\paragraph*{Multi-Agent Graph Exploration on general graphs and trees} In the case of multi-agent graph exploration the main focus of the literature lies on trees. For $k$ agents, Fraigniaud et al.\ gave an $O(\frac{k}{\log k})$-competitive algorithm
for the exploration of trees \cite{Fraigniaud2006}. For the case of unrestricted communication between agents this was improved very recently by Cosson and Massouli{\'e} \cite{DBLP:conf/innovations/CossonM24}, who developed an $O(\sqrt{k})$ competitive tree exploration algorithm.
In the case when $k = 2^{\omega(\sqrt{\log D \log \log D})}$ and $n = 2^{O(2^{\sqrt{\log D}})}$ (where $D$ is the diameter of the tree and $n$ the number of nodes) Ortolf and Schindelhauer presented a subpolynomial $k^{o(1)}$-competitive
algorithm \cite{Ortolf2014}.
Fraigniaud et al.\ also provided a lower bound of $\Omega(2-\frac{1}{k})$ even for the case of global communication, which was later improved to $\Omega(\frac{\log k}{\log \log k})$ by Dynia et al.\ 
\cite{Dynia}. Dynia et al.\ \cite{Dynia2006} looked at a different cost model, in which not the general exploration time, but the distance traversed by any
individual agent has to be minimized. They developed a $4-\frac{2}{k}$-competitive algorithm for the exploration of trees 
 and gave a lower bound of $1.5$. Dynia et al.\ \cite{Dynia2006a} have analyzed the relationship between height and density of a tree and the competitive 
ratio of exploration strategies. They provided an algorithm with a competitive ratio of $O(D^{(1-\frac{1}{p})})$ where $p$ describes the density of a given tree. Dereniowski et al.\ gave an exploration strategy for general graphs using $Dn^{1+\Omega(1)}$ agents, that explores a graph in $O(D)$ time steps,
which implies a constant competitive ratio. \cite{Dereniowski2015}. This was extended on trees by Disser et al.\ to show a lower bound of $\omega(1)$ where $n \log^c n \leq k \leq Dn^{1+o(1)}$ \cite{Disser2020}.

\paragraph*{Multi-Agent Graph Exploration on restricted graph classes} Higashikawa et al.\ gave an optimal $1.5$-competitive algorithm for cycles using global communication and a lower bound of $\Omega(\frac{k}{\log k})$ for trees when using greedy algorithms, 
which implies that the $\frac{k}{\log k}$-competitive algorithm by Fraigniaud et al.\ cannot be improved by using greedy strategies \cite{Higashikawa2012}. 
Ortolf and Schindelhauer gave an upper bound of $O(\log^2 n)$ for $k$-agent exploration of $n \times n$ grid graphs with rectangular holes when using local communication
and a lower bound of $\Omega(\frac{\log k}{\log \log k})$ for deterministic as well as $\Omega(\sqrt{\frac{\log k}{\log \log k}})$ for randomized strategies \cite{Ortolf2012}.

\paragraph*{Other Graph Exploration variations} Aside from the previous examples, there exist different variations of multi-agent graph exploration, like exploration of unlabeled graphs using devices called pebbles \cite{Bender2002, Disser2018},
exploration using battery constrained agents \cite{Bampas2018}, exploration where edges are opaque \cite{Brass2011}, or exploration where nodes are space restricted and can only be occupied by one agent at any given time \cite{Czyzowicz2017}.
A $2$-competitive algorithm for exploration of cycles was developed by Osula for a model in which the exploration is started without any agents and spawning a new agent incurs an invocation cost $q$ \cite{Osula2017}.

\subsection{Outline}\label{outline}
In \S \ref{preliminaries} we define the exploration and cost models used. 
In addition, we describe the selected graph classes and some general concepts of online computation.\\
In \S \ref{ale_tad}, we investigate the competitiveness of the \textit{ALE algorithm} on tadpole graphs.
In \S \ref{cycles} to \ref{ntadpoles} we provide a modification of the \textit{ALE algorithm} and analyze its competitive ratio for the different graph classes, starting with cycles
in \S \ref{cycles}, followed by strategies for tadpole graphs using two to four agents in \S \ref{tadpoles}. Finally, a discussion of the results is given in \S \ref{discussion}.
\section{Preliminaries} \label{preliminaries}

In this section, we first describe the underlying graph and graph exploration models in \S\ref{graph_classes} and \S\ref{graph_exploration} respectively.
After that we briefly describe the general idea of competitive analysis and online algorithms in \S\ref{competitive}.

\subsection{Graph Classes}\label{graph_classes}

We consider undirected weighted Graphs, $G=(V,E)$, where edge weights
will be denoted as $l(e)$ for an edge $e\in E$. Edge weights are nonzero
and nonnegative. We call the weight of an edge its length. We use
$d(v,w)$ to describe the shortest path between two nodes $v,w \in V$. The length
of a path is the sum of the lengths of all edges on that path. We consider
two graph classes: cycles and tadpole graphs, which are defined as followed:

\paragraph*{Cycle} A cycle $C=(V,E)$ is a connected graph, where each node in $V$ has degree 2.
We call $L = \sum_{e \in E}^{}l(e)$ the length of the cycle.

\paragraph*{Tadpole Graph} A tadpole graph $T=(V,E)$ is a graph, which consists of a cycle to which
a tail is attached. In a tadpole graph we denote the length of the cycle as $L_c$ and the length
of the tail with $L_t$. We call the cycle node with degree 3, to which the tail is attached, the intersection $v_i$.
The single node with degree 1 on the tail will be called the end of the tail $v_t$.

\subsection{Online Graph Exploration}\label{graph_exploration}
In the online graph exploration problem, all nodes of a given weighted graph $G=(V,E)$ have to be visited by an agent $a$, starting and ending on a specific node $s \in V$.
The graph is initially unknown, except for the starting node $s$ and its neighborhood $N(s)$. Whenever the agent traverses an edge to a yet unvisited node $v$, all
adjacent edges, as well as $N(v)$ is revealed. Nodes have labels and can be distinguished, but the labels do not provide any further information about the graph to the agent.
When the agent is located on some node $v$ and the exploration is not finished, it can either choose an edge $e$ adjacent to $v$, which it will traverse in $l(e)$ time,
or it can choose to wait on the current node while time progresses.\\
We consider graph exploration with $k \in \mathbb{N}$ agents $a_1,...,a_k$, for which we assume unlimited computational power and shared knowledge, meaning that as soon as any agent learns
something about the neighborhood of a node, all other agents instantly receive the same information.\\
A graph is considered explored, when each node in $V$ has been visited by any agent, and
all $k$ agents have returned to the starting node $s$. All agents move at identical speed, taking a time of $l(e)$ for traversing some edge $e$. When the exploration is finished, the time passed is denoted as $\mathcal{T}$, while the distance traversed by the
agent $a_i$ is denoted as $d(a_i)$. The goal of a graph exploration strategy, is to minimize the cost of the exploration. We consider two models for the cost of an exploration:

\paragraph*{Time Model} 
In the time model the cost of the exploration is the entire exploration time $\mathcal{T}$ taken until the last agent returns to $s$ after all nodes of the graph have been visited.

\paragraph*{Energy Model}
In the energy model the cost of the exploration is the maximum distance \linebreak$\max(d(a_1), ... ,d(a_k))$ traversed by any agent during the exploration of the graph.

\begin{figure}[H]
    \begin{center}
        \begin{tikzpicture}[node distance={15mm}, thick, main/.style = {draw, circle}] 
            \node[main, minimum size=0.5cm] (2) {}; 
            \node[main, minimum size=0.5cm] (3) [above right of = 2, above=-1.5em] {$_{v_i}$ }; 
            \node[main, minimum size=0.5cm] (4) [above left of = 3, above=-1.5em] {};
            \node[main, fill=gray!50, minimum size=0.5cm] (6) [right of = 3, right=-0.5em] {$_{s}$ };
            \node[main, minimum size=0.5cm] (7) [left of = 2, left=-0.5em] {$_{v_s}$ };
            \node[main, minimum size=0.5cm] (8) [left of = 4, left=-0.5em] {$_{v_\ell}$};
            \node[main, minimum size=0.5cm] (9) [right of = 6, right=-0.5em] {$_{v_t}$ };
            \draw (7) [bend left] to node[left] {$e_{mid}$} (8);
            \draw (2) [bend right] to node[right] {} (3);
            \draw (3) [bend right] to node[above right] {} (4);
            \draw (3) [] to node[above] {} (6);
            \draw (2) [] to node[below] {} (7);
            \draw (4) [] to node[above] {} (8);
            \path [red, thick, dashed, ->] (3) edge[out=90,in=45] node[above] {{$d_\ell$}} (8);
            \path [blue, thick, dashed, ->] (3) edge[out=270,in=315] node[below right] {{$d_s$}} (7);
            \path [olive, thick, dashed, ->] (6) edge[out=15,in=165] node[above] {{$d_t$}} (9);
            \path [violet, thick, dashed, ->] (6) edge[out=165,in=15] node[above] {{$d_i$}} (3);
            \draw (6) [] to node[above] {} (9);
        \end{tikzpicture} 
    \end{center}
    \caption[Tadpole example]{Tadpole Graph Example. The paths $p_\ell, p_s, p_i, p_t$ and the edges $e_{mid}$ and $e_{max}$ are~marked.}
    \label{Example}
\end{figure}
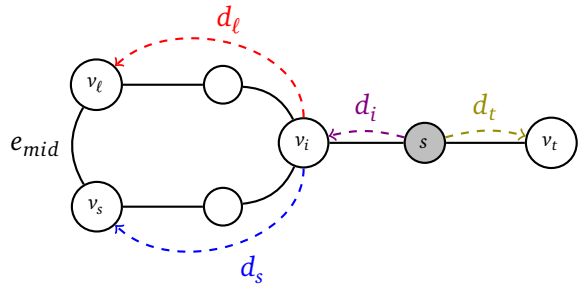

For any tadpole graph, we denote the longest edge of the cycle as $e_{max}$. If multiple edges share the highest length, one of these edges can be picked at random
to be $e_{max}$. We call the point on the cycle with exactly distance $L_c/2$ in both directions to either
the starting point $s$ (when $s$ in on the cycle), or the intersection $v_i$ (when $s$ is on the tail), the midpoint $m$ of the cycle. If $m$ falls onto an edge, this edge is
called $e_{mid}$. If $m$ falls onto a node, the node is called $v_{mid}$. In this case we consider $l(e_{mid}) = 0$.
We consider $e_{mid}=(v_l, v_s)$. If $e_{mid}$ does
not exist, we consider $v_l = v_s = v_{mid}$. 
We denote $p_l$ and $p_s$ with length $d_l$ and $d_s$ as edge-disjoint paths from $s$ (if $s$ is located on the cycle) or $d_i$ (if $s$ is located on the tail) to $v_l$ and $v_s$ respectively, 
where $p_l$ and $p_s$ combined contain all edges of the cycle except $(v_l, v_s)$. We always assume w.l.o.g that $d_s \leq d_l$.\\
The path $p_i$ with length~$d_i$ always denotes the shortest path from $s$ to $v_i$. 
The path $p_t$ with length~$d_t$ denotes the path from $v_i$ to $v_t$ when $s$ is located on the cycle. When $s$ is located on the tail, $p_t$ denotes the path from $s$
to $v_t$. An example of a tadpole graph and the labeled paths and edges is given in Figure \ref{Example}.

\subsection{Competitive Analysis}\label{competitive}
Algorithms where the input is revealed piece by piece instead of being known from the start are called \textit{online algorithms}. For the assessment of online algorithms,
their result gets compared to the optimal result that can be achieved by an algorithm which knows the whole input from the start. In our case, we compare the online
exploration strategies to the optimal paths that can be found when the entire graph is already known. For this we use the definition also used by Brandt et al.~\cite{Brandt2020}, but extend it
by adding the number of agents $k$ to the problem instance.

\paragraph*{Competitive Ratio} Let $opt(G,s,k)$ be the cost of an optimal (offline) solution
for a graph exploration instance, while $\onl_{A}(G,s,k)$ is the cost of a solution found by some online algorithm. We call $\frac{\onl_{A}(G,s,k)}{opt(G,s,k)}$ the \textit{competitive ratio}
for this given instance. Given a graph class $\mathcal{G}$ and a number of agents $k \in \mathbb{N}$, we say that an Algorithm $A$ is $c$-competitive for $\mathcal{G}$ and $k$, if $c$ is the supremum for
the competitive ratio of $A$ in all possible graph exploration instances $(G =(V,E) \in \mathcal{G}, s \in V, k)$. Note that by this definition, an algorithm cannot be less than $1$-competitive.\\
If an exploration strategy is
$c$-competitive for the time model, it is at most $c$-competitive for the energy model, since in an exploration that takes some time $\mathcal{T}$, no agent can traverse a distance further than $\mathcal{T}$,
while in the offline case the maximum distance traversed by any agent matches the general exploration time.
\section{ALE on Tadpole Graphs} \label{ale_tad}

In this section, we extend the \textit{ALE (Avoid-Longest-Edge) Algorithm} by Higashikawa et al.\ \cite{Higashikawa2012} to a strategy for tadpole graph exploration.
The \textit{ALE Algorithm} uses two agents to explore a cycle, by sending the agents clockwise and counterclockwise respectively,
moving only the agent that sees the edge with lower weight at a time. If both agents see edges with the same weight, the edge that is traversed is chosen at random. 
When all nodes have been visited, both agents have full knowledge of the graph and return to the start via their shortest paths. The \textit{ALE Algorithm} reaches a 
competitive ratio of $1.5$ on cycles in the time model.
We first show a competitive lower bound of $2$ for the exploration of tadpole graphs with $k\geq3$ agents for any \textit{ALE}-based exploration strategy.
We then prove that an exploration strategy based on the \textit{ALE Algorithm} is $3$-competitive when using $k=3$ agents and $2$-competitive when using $k=4$ agents 
on tadpole graphs in the time model. 

\paragraph*{Lower Bound}

We show that when starting on a node of the cycle, an adversary can
always hide the tail of the tadpole graph behind the initially seen, but most expensive edge of the tadpole graph, leading to a competitive lower bound
of $2$ for any \textit{ALE}-based exploration of tadpole graphs using at least $k\geq3$~agents.

\begin{theorem}[Lower bound of ALE]\label{AleLower}
    Given $k\geq3$ agents, a strategy that tries to avoid the longest edge of a given graph cannot have a competitive ratio of 
    less than two for the class of tadpole graphs.
\end{theorem}

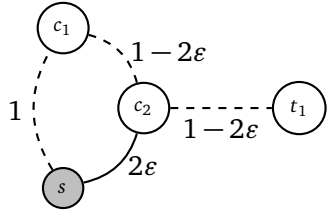
\begin{figure}[H]
    \begin{center}
    \begin{tikzpicture}[node distance={15mm}, thick, main/.style = {draw, circle}, scale=0.5] 
        \node[main, fill=gray!50, minimum size=0.5cm] (2) {$_{s}$ }; 
        \node[main, minimum size=0.5cm] (3) [above right of = 2] {$_{c_2}$ }; 
        \node[main, minimum size=0.5cm] (4) [above left of = 3] {$_{c_1}$ };
        \node[main, minimum size=0.5cm] (6) [right of = 3, right=0.5em] {$_{t_1}$ };
        \draw (2) [bend left, dashed] to node[left] {$1$} (4);
        \draw (2) [bend right] to node[right] {$2\varepsilon$} (3);
        \draw (3) [bend right, dashed] to node[right] {$1 - 2\varepsilon$} (4);
        \draw (3) [dashed] to node[below] {$1 - 2\varepsilon$} (6);
    \end{tikzpicture} 
    \end{center}
    \caption[ALE is not $1.5$-competitive]{A graph in which ALE has a competitive ratio of $2$ for the time model. The dashed lines
    are paths, consisting entirely of edges with length $\varepsilon$.}
    \label{ALETad}
\end{figure}

\begin{proof}
    Consider how an optimal offline strategy, when given $k\geq3$ agents, takes time $2$ to explore the graph shown in Figure \ref{ALETad}, by sending two
    agents clockwise and counterclockwise to $c_1$ and back, as well as one agent on the shortest path to $t_1$ and back.\\
    Since the edge $(s,c_2)$ is the longest edge in the graph, any ALE strategy will send some number of agents clockwise around the graph, from
    $s$ to $c_2$. At this point, after $2-2\varepsilon$ time units have passed, the cycle is explored and the shortest way to complete the exploration, is to send
    some agent from $c_2$ to $t_1$ and back to $s$, leading to an exploration time of $4-4\varepsilon$ and a competitive ratio of $\frac{4-4\varepsilon}{2}$.
    The supremum for the competitive ratio on any graph like the one shown in Figure \ref*{ALETad} is $2$, having arbitrary small $\varepsilon>0$.
\end{proof}

\paragraph*{Upper Bound}

Having shown the lower bound of $2$ for $k\geq 3$ agents, we now show an exploration strategy that reaches this competitive ratio with a team of four agents,
while reaching a competitive ratio of $3$ with a team of three agents. The proof for this Theorem \ref{AleUpper} has been moved to the \hyperref[appendix]{Appendix}.

\begin{theorem}[ALE on tadpole graphs]\label{AleUpper}
    On tadpole graphs, using the ALE strategy, a team of three agents reaches a competitive ratio of $3$,
     while a team of four agents reaches the optimal competitive ratio of $2$ for the time model.
\end{theorem}

\section{Cycles in the Energy Model} \label{cycles}

In this section, we show that two agents suffice to obtain a $1$-competitive algorithm for the exploration
of cycles in the energy model. For this we modify the \textit{ALE Algorithm} by Higashikawa et al.\ \cite{Higashikawa2012}. We first show 
that the \textit{ALE Algorithm} has a competitive ratio greater than $1$ in the energy model. After that we construct a modified
algorithm and prove that it reaches $1$-competitiveness in the energy model.

\subsection{ALE is $1.5$-competitive in the energy model}

\begin{theorem}[The \textit{ALE} algorithm and the energy model]\label{AleLowerTad}
    The \textit{ALE} algorithm has a competitive ratio of $1.5$ in the energy model on cycles.
\end{theorem}

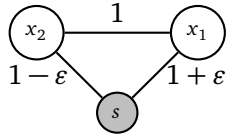
\begin{figure}[H]
    \begin{center}
    \begin{tikzpicture}[node distance={15mm}, thick, main/.style = {draw, circle}, scale=0.5] 
        \node[main, fill=gray!50, minimum size=0.5cm] (2) {$_{s}$ }; 
        \node[main, minimum size=0.5cm] (1) [above left of = 2] {$_{x_2}$ }; 
        \node[main, minimum size=0.5cm] (3) [above right of = 2] {$_{x_1}$ }; 
        \draw (1) to node[left] {$1-\varepsilon$} (2);
        \draw (2) to node[right] {$1+\varepsilon$} (3);
        \draw (3) to node[above] {$1$} (1);
    \end{tikzpicture} 
    \end{center}
    \caption[ALE is not $1$-competitive]{Since $(s,x_1)$ is the longest edge, one agent moves clockwise from $s$ to $x_1$.}\label{ale}
\end{figure}

\begin{proof}
    Let $0 < \varepsilon < 0.5$. Consider the graph shown in Figure \ref{ale}. Since $(s,x_1)$ is the longest edge in the graph, one agent explores the graph clockwise from $s$ to $x_1$. After the agent
    reaches $x_1$ the shortest path to $s$ is the edge $(x_1, s)$, leading to a cost of $3$. An optimal strategy would have sent both agents from $s$ to $x_1$ 
    and $x_2$ respectively, leading to an exploration cost of $2(1+\varepsilon)$. The overhead of ALE in this case is $\frac{3}{2(1+\varepsilon)}$, leading to a competitive ratio of
    $1.5$ for arbitrary small $\varepsilon$.
\end{proof}

\subsection{The AMP Algorithm}

To construct a $1$-competitive algorithm for the energy model we make use of an observation made by Higashikawa et al.\ concerning the optimal offline exploration 
of a cycle \cite{Higashikawa2012}. Given a cycle graph $C=(V,E)$, a starting node $s \in V$, and the midpoint of the cycle $m$. Then there are two cases:

\begin{enumerate}
    \item $m$ falls onto a node $v_{mid} \in V$ 
    \item $m$ falls onto an edge $e_{mid} = (v_\ell, v_s) \in E$ 
\end{enumerate}

In case 1 the optimal offline strategy sends two agents clockwise and counterclockwise to $v_{mid}$ and back, leading to an exploration time of exactly $L$. In case 2 the agents
are sent to $v_\ell$ and $v_s$ and back, while the edge $e_{mid}$ is not traversed by any agent. The following algorithm sends the agents around the cycle in opposite directions,
with only one moving agent at a time, like the \textit{ALE Algorithm} does, but instead of choosing the locally cheaper edge in each step, the edge which minimizes the distance
traversed by any agent up to that point is chosen. We show that with this change the agents traverse exactly the same paths as the offline optimal strategy, leading to a competitive ratio
of $1$ for the energy model while still achieving the optimal competitive ratio of $1.5$ for the time model.

    \begin{algorithm}[H]
        \caption{AMP (Avoid Midpoint)}\label{alg:cap}
        \begin{algorithmic}
        \small
        \Require A unknown cycle graph $G=(V,E)$, Two agents $a_1, a_2$, a starting node $s \in V$
        \State $d(a_1) \gets 0; d(a_2) \gets 0$ \Comment{Agents track their already traversed distance}
        \State $Exp \gets \{s\}$ \Comment{Keep track of already explored vertices}
        \State Let $n(a_1)$ and $n(a_2)$ describe the next nodes seen by the agents
        \State Assign the two neighbors of $s$ randomly to $n(a_1)$ and $n(a_2)$
        \While{$n(a_1) \not\in Exp \lor n(a_2) \not\in Exp$} \Comment{While graph is not explored}
        \If{$d(a_1) + l(n(a_1)) < d(a_2) + l(n(a_2))$}  
            \State $a_1$ traverses edge to $n(a_1)$ 
            \State $Exp \gets Exp \cup \{n(a_1)\}$
            \State $d(a_1) \gets d(a_1) + l(n(a_1))$
            \State $n(a_1) \gets$ next revealed node
        \Else 
            \State $a_2$ traverses edge to $n(a_2)$
            \State $Exp \gets Exp \cup \{n(a_2)\}$
            \State $d(a_2) \gets d(a_2) + l(n(a_2))$
            \State $n(a_2) \gets$ next revealed node
        \EndIf
        \EndWhile
        \State $a_1$ and $a_2$ return to $s$ using their shortest paths.
        \end{algorithmic}
        \end{algorithm}

\begin{theorem}[Cycles in the Energy Model]\label{ampCycles}
    For the energy model, the \textit{AMP (Avoid Midpoint)} Algorithm \ref{alg:cap} explores a cycle with
    a competitive ratio of $1$. For the time model the algorithm explores a cycle with
    a competitive ratio of $1.5$.
\end{theorem}

\begin{proof}
    We can prove the $1$-competitiveness of the \textit{AMP algorithm} for the energy model, by showing that
    the agents only traverse the cycle to $v_s$ and $v_\ell$ (or $v_{mid}$) before returning, and never traverse the edge $e_{mid}$. We then use this result to prove
    the $1.5$-competitive ratio for the time model.

    \paragraph*{Energy Model} If both agents reach $v_{mid}$ at the same time, all nodes have been visited and both agents have traversed a distance of $L/2$ before and $L$ after backtracking, which matches 
    the result of the optimal offline strategy.\\
     If an agent $a_1$ reaches $v_{mid}$ first, having traversed a distance $d_1 = L/2$, the other agent $a_2$ must have traversed a distance $d_2 < d_1$. Thus, $a_2$ moves to $v_{mid}$ without stopping.
     Then both agents have traversed a distance of $L/2$ when all nodes have been visited and $L$ after backtracking.\\
    Assume w.l.o.g an agent $a_1$ is currently located on $v_s$ after having traversed distance $d_s$ and has not yet traversed $e_{mid}$, while the other agent $a_2$ has not reached $v_\ell$ yet.
    Since $d_\ell < L/2 < d_s + l(e_{mid})$ the agent $a_2$ traverses to $v_\ell$ without stopping. At this point all nodes have been visited and the agents backtrack, having taken the same paths
    as they would have in an optimal offline exploration.

    \paragraph*{Time Model} Recall that an optimal strategy takes time $2d_\ell$. In the AMP algorithm, the agents traverse the distances $d_s$ and $d_\ell$. 
    Since only one agent is traversing an edge at a time until all nodes have been visited, the maximum 
    time needed to visit all nodes is $d_s + d_\ell \leq 2d_\ell$. After backtracking to $s$ the full exploration time is at most $3d_\ell$, leading to a competitive ratio of $1.5$.
\end{proof}
\section{AMP on Tadpole Graphs} \label{tadpoles}
In this section we cover multi-agent graph exploration on tadpole graphs. In \S\ref{lower_bound_tad} we show a lower bound of $1.5$ for the online exploration of tadpole graphs with $k\geq2$ agents in the time-, and $k=2$ agents in the energy model.
In \S\ref{tadpole_two} we show that random choice with two agents leads to a $2.5$-competitive strategy. In \S\ref{tadpole_three} we provide a modification of the \textit{AMP Algorithm} from \S \ref{cycles} that yields a competitive ratio of $1$ for the energy- and $2$ for the time model, 
using three agents. In \S\ref{tadpole_four} we describe an exploration strategy which matches the competitive lower bound of $1.5$ using $k=4$ agents. In Table \ref{table:tadpole} an overview of the results of this section is given.

    \begin{table}[H]
        \small
        \centering
        \begin{tabular}{|l|ll|ll|}
        \hline
                  & \multicolumn{2}{c|}{Time Model}                                     & \multicolumn{2}{c|}{Energy Model}                                   \\ \hline
                  & \multicolumn{1}{c|}{Lower Bound} & \multicolumn{1}{c|}{Upper Bound} & \multicolumn{1}{c|}{Lower Bound} & \multicolumn{1}{c|}{Upper Bound} \\ \hline
        1 Agent   & \multicolumn{1}{l|}{2 \cite{Brandt2020}}           & 2\cite{Brandt2020}                                & \multicolumn{1}{l|}{2\cite{Brandt2020}}           & 2\cite{Brandt2020}                               \\ \hline
        2 Agents  & \multicolumn{1}{l|}{1.5 (Thm. \ref{TadpoleTimeLower})}         & 2.5 (Thm. \ref{TadpoleUpperTwo})                              & \multicolumn{1}{l|}{1.5 (Thm. \ref{TadpoleEnergyLower})}            & 2.5 (Thm. \ref{TadpoleUpperTwo})                               \\ \hline
        3 Agents  & \multicolumn{1}{l|}{1.5 (Thm. \ref{TadpoleTimeLower})}         & 2 (Thm. \ref{TadpoleUpperThree})                               & \multicolumn{1}{l|}{1}           & 1 (Thm. \ref{TadpoleUpperThree})                               \\ \hline
        4+ Agents & \multicolumn{1}{l|}{1.5 (Thm. \ref{TadpoleTimeLower})}         & 1.5 (Thm. \ref{TadpoleUpperFour})                             & \multicolumn{1}{l|}{1}           & 1 (Thm. \ref{TadpoleUpperThree})                               \\ \hline
        \end{tabular}
        \caption{Results for multi-agent exploration of tadpole graphs.} \label{table:tadpole}
    \end{table}

\subsection{A lower bound for tadpole graphs}\label{lower_bound_tad}

In this section we show lower bounds for graph exploration on tadpole graphs. We first show that the exploration of any tadpole graph has a competitive ratio of at least
$1.5$ for $k\geq 2$ agents in the time model. For the energy model we show a lower bound of $1.5$ with $k=2$ agents. 
\subsubsection*{$k\geq2$ agents in the time model} We modify the graph construction Higashikawa et al.\ used to show the lower bound of $1.5$ for cycles with $2$ agents \cite{Higashikawa2012},
to show a lower bound of $1.5$ for the time model for tadpole graphs with $k \geq 2$ agents, by simply attaching a tail consisting of a single edge with of length $\varepsilon$ to the starting node $s$.
Note that the analysis stays basically identical to the analysis by Higashikawa et al.\ \cite{Higashikawa2012}, since in the offline case the tail can always be explored in parallel, having no impact on the
optimal exploration time.

\begin{theorem}[Time Model: Lower bound for exploring tadpole graphs]\label{TadpoleTimeLower}
    For any number of agents $k \geq 2$, any online exploration strategy for tadpole graphs has a competitive ratio of at least $1.5$ in the time model.
\end{theorem}

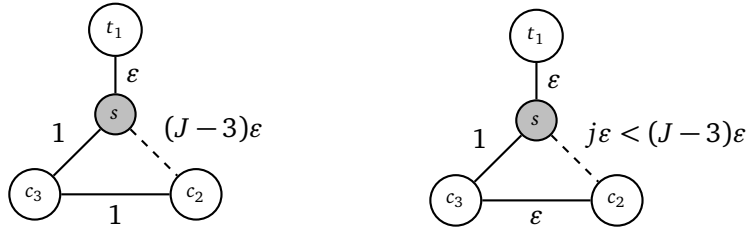
\begin{figure}[H]
    \begin{center}
    \begin{tikzpicture}[node distance={15mm}, thick, main/.style = {draw, circle}, scale=0.5] 
        \node[main, fill=gray!50, minimum size=0.5cm] (3) [] {$_{s}$ }; 
        \node[main, minimum size=0.5cm] (2) [below left of = 3] {$_{c_3}$ }; 
        \node[main, minimum size=0.5cm] (4) [below right of = 3] {$_{c_2}$ };
        \node[main, minimum size=0.5cm] (6) [above of = 3, above=-2em] {$_{t_1}$ };
        \draw (2) [] to node[below] {$1$} (4);
        \draw (2) [] to node[above left] {$1$} (3);
        \draw (3) [dashed] to node[above right] {$(J-3)\varepsilon$} (4);
        \draw (3) [] to node[right] {$\varepsilon$} (6);
    \end{tikzpicture} 
    \hspace*{5em}
    \begin{tikzpicture}[node distance={15mm}, thick, main/.style = {draw, circle}, scale=0.5] 
        \node[main, fill=gray!50, minimum size=0.5cm] (3) [] {$_{s}$ }; 
        \node[main, minimum size=0.5cm] (2) [below left of = 3] {$_{c_3}$ }; 
        \node[main, minimum size=0.5cm] (4) [below right of = 3] {$_{c_2}$ };
        \node[main, minimum size=0.5cm] (6) [above of = 3, above=-2em] {$_{t_1}$ };
        \draw (2) [] to node[below] {$\varepsilon$} (4);
        \draw (2) [] to node[above left] {$1$} (3);
        \draw (3) [dashed] to node[above right] {$j\varepsilon < (J-3)\varepsilon$} (4);
        \draw (3) [] to node[right] {$\varepsilon$} (6);
    \end{tikzpicture} 
    \end{center}
    \caption[Tadpole lower bound]{Any number of agents $k\geq2$ cannot explore the given graph with a competitive ratio lower than $1.5$. The dotted line is a path consisting entirely of edges with length $\varepsilon$.
    Case 1 (Left): the graph an adversary reveals if $(s, c_3)$ is not traversed after an agent traversed a distance of $(J-3)\varepsilon$ on the path between $s$ and $c_2$.
    Case 2 (Right): the graph an adversary creates when $(s,c_3)$ is traversed before an agent traverses a distance of $j\varepsilon<(J-3)\varepsilon$ on the path between $s$ and $c_2$.}
    \label{tadpole-time}
\end{figure}
\begin{proof}
    Consider the graphs shown in Figure \ref{tadpole-time}. In the beginning, the agents only see two edges of length $\varepsilon$ and one edge of length $1$. There are two cases depending on when an agent decides to 
    explore the edge $(s,c_3)$. Let $\varepsilon=\frac{1}{J}$ and $J \in \mathbb{N}$.
    \begin{enumerate}
        \item An agent traverses a distance of $(J-3)\varepsilon$ on $(s,c_2)$ before an agent enters the edge $(s,c_3)$.
        \item An agent starts traversing $(s,c_3)$ when a distance $j\varepsilon<(J-3)\varepsilon$ is traversed on $(s,c_2)$
    \end{enumerate}
    \paragraph*{Case 1} In this case the adversary sets $l(c_2, c_3)=1$. The optimal offline strategy would have sent one agent from $s$ to $c_3$ and back, while another agent in parallel explores the edge
    $(s,t_1)$ and the path $(s,c_2)$, resulting in an exploration time of $2$. The online strategy has already taken a time of at least $(1-3\varepsilon)$ and needs at least time $2$ to complete the exploration,
    leading to a minimum exploration time of $(3-3\varepsilon)$.
    \paragraph*{Case 2} In this case the adversary sets $l(c_2, c_3)=\varepsilon$ and lets $c_2$ be the next visited node by the searcher on the path from $s$ to $c_2$. 
    The optimal offline strategy would send one agent from $s$ to $c_3$ via $c_2$ and back, while another agent explores the tail in parallel, leading to an exploration time of $2(j+1)\varepsilon$. The online strategy
    has already traversed a distance of at least $(j-1)\varepsilon$ before an agent started traversing $(s,c_3)$. After reaching $c_3$ the fastest way to $s$ is traversing the graph via $c_2$. Since
    $1=J\varepsilon\geq (j+3)\varepsilon$, the entire exploration
    has at least length $(j-1)\varepsilon + (j+3)\varepsilon + (j+1)\varepsilon = 3(j+1)\varepsilon$.
\end{proof}
\subsubsection*{Two agents in the energy model} For showing $1.5$-competitiveness in the energy model, we modify the construction Dynia et al.\ used to show a lower bound of $1.5$ for the exploration
of trees \cite{Dynia2006}, to create a tadpole graph.

\begin{theorem}[Energy Model: Lower bound for exploring tadpole graphs]\label{TadpoleEnergyLower} 
    For $k=2$ agents, no online exploration strategy on tadpole graphs can have a competitive ratio better than $1.5$ in the energy model.
\end{theorem}

\begin{figure}[H]
    \begin{center}
    \begin{tikzpicture}[node distance={15mm}, thick, main/.style = {draw, circle}, scale=0.5] 
        \node[main, fill=gray!50, minimum size=0.5cm] (2) {$_{s}$ }; 
        \node[main, minimum size=0.5cm] (1) [above left of = 2] {$_{c_1}$ }; 
        \node[main, minimum size=0.5cm] (4) [left of = 1] {$_{c_2}$ }; 
        \node[main, minimum size=0.5cm] (3) [below left of = 2] {$_{c_3}$ };
        \node[main, minimum size=0.5cm] (5) [left of = 3] {$_{c_4}$ };  
        \node[main, minimum size=0.5cm] (6) [right of = 2] {$_{t_1}$ };
        \node[main, minimum size=0.5cm] (7) [right of = 6] {$_{t_2}$ };
        \draw (1) to node[above right] {$1-\varepsilon$} (2);
        \draw (2) to node[below right] {$1-\varepsilon$} (3);
        \draw (4) [dashed] to node[left] {$\infty$} (5);
        \draw (1) to node[above] {$\varepsilon$} (4);
        \draw (3) to node[below] {$\varepsilon$} (5);
        \draw (2) to node[above] {$1-\varepsilon$} (6);
        \draw (6) to node[above] {$1+\varepsilon$} (7);
    \end{tikzpicture} 
    \end{center}
    \caption[Lower Bound Energy Model]{Lower Bound construction for the energy model.}
    \label{fig:energybound}
\end{figure}
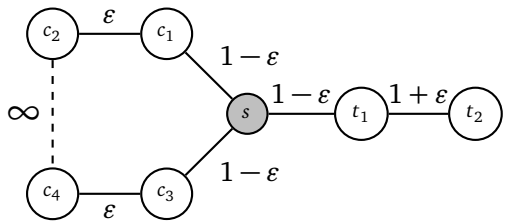

\begin{proof}
    Let $\varepsilon > 0$ be an arbitrary small and $\infty$ be an arbitrary large number.
    Consider the graph shown in Figure \ref{fig:energybound}. An optimal offline solution starting at $s$ would send one agent from $s$ to $t_2$ and back, while the other agent
    simultaneously explores the cycle, traversing each edge except $(c_2, c_4)$ twice, leading to a maximum distance of $4$ traversed by any agent.\\
    In the online case the agents cannot distinguish the edges they see when starting at $s$. Since $c_2$ and $c_4$ are hidden behind edges of length $\varepsilon$, the remaining two
    paths cannot be distinguished even if one agent has already traversed a path from $s$ to $c_2$ or $c_4$. Because of this an adversary could always affect the exploration
    order of the two agents, so that the maximum distance traversed by any agent is at least $6-2\varepsilon$. This leads to a competitive ratio of $1.5$ for arbitrary small $\varepsilon$.
\end{proof}

\subsection{Tadpole graphs with two agents}\label{tadpole_two}

We construct a $2.5$ competitive exploration strategy by using two agents for the exploration of a tadpole graph, where random paths are chosen as soon as the intersection is found.

    \begin{theorem}[Upper bound for two-agent randomized tadpole graph exploration]\label{TadpoleUpperTwo}
        Using the \textit{AMP} strategy with two agents for the exploration of tadpole graphs, choosing random paths as soon as the intersection is found,
        leads to a competitive upper bound of $2.5$ compared to an optimal offline strategy using two agents.
    \end{theorem}

\begin{proof} Consider how an online strategy exploring tadpole graphs can distinguish 3 cases for the start of an exploration: If the degree of the starting node $s$ is 1,
    it must be the end of the tail, if the degree is 3, it must be the intersection $v_i$. If starting node has degree 2, it can be any other node of the graph.
    The possible starting positions for an exploration are depicted in Figure \ref{Tail}.
    We describe the exploration strategies for each of these starting positions and show that for any of those, an online strategy choosing
    a random direction when finding the intersection, takes at most time $5\max(d(s,v), v\in \{v_\ell, v_t\})$ for the exploration of any tadpole graph. 
    Since any optimal strategy takes at least time $2\max(d(s,v), v\in \{v_\ell, v_t\})$ to explore any tadpole graph, this directly leads to a competitive
    upper bound of $2.5$. We then provide an example, in which our strategy takes $2.5$ times the optimal exploration time.

    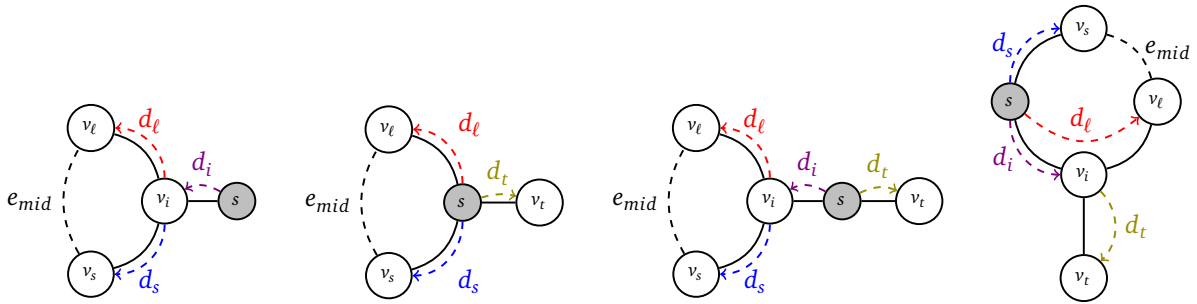
\begin{figure}[H]
        \begin{center}
        \adjustbox{max width=\textwidth}{
        \begin{tikzpicture}[node distance={15mm}, thick, main/.style = {draw, circle}, scale=0.3]
            \node[main, minimum size=0.5cm] (2) {$_{v_s}$ }; 
            \node[main, minimum size=0.5cm] (3) [above right of = 2] {$_{v_i}$ }; 
            \node[main, minimum size=0.5cm] (4) [above left of = 3] {$_{v_\ell}$ };
            \node[main, fill=gray!50, minimum size=0.5cm] (6) [right of = 3, right=-2em] {$_{s}$ };
            \draw (2) [bend left, dashed] to node[left] {$e_{mid}$} (4);
            \draw (2) [bend right] to node[right] {} (3);
            \draw (3) [bend right] to node[above right] {} (4);
            \draw (3) [] to node[above] {} (6);
            \path [red, thick, dashed, ->] (3) edge[out=90,in=0] node[above] {{$d_\ell$}} (4);
            \path [blue, thick, dashed, ->] (3) edge[out=270,in=0] node[below] {{$d_s$}} (2);
            \path [violet, thick, dashed, ->] (6) edge[out=150,in=30] node[above] {{$d_i$}} (3);
        \end{tikzpicture} 
        \hspace*{1em}
        \begin{tikzpicture}[node distance={15mm}, thick, main/.style = {draw, circle}, scale=0.3] 
            \node[main, minimum size=0.5cm] (2) {$_{v_s}$ }; 
            \node[main, fill=gray!50, minimum size=0.5cm] (3) [above right of = 2] {$_{s}$ }; 
            \node[main, minimum size=0.5cm] (4) [above left of = 3] {$_{v_\ell}$ };
            \node[main, minimum size=0.5cm] (6) [right of = 3, right=-2em] {$_{v_t}$ };
            \draw (2) [bend left, dashed] to node[left] {$e_{mid}$} (4);
            \draw (2) [bend right] to node[right] {} (3);
            \draw (3) [bend right] to node[above right] {} (4);
            \draw (3) [] to node[above] {} (6);
            \path [red, thick, dashed, ->] (3) edge[out=90,in=0] node[above right] {{$d_\ell$}} (4);
            \path [blue, thick, dashed, ->] (3) edge[out=270,in=0] node[below right] {{$d_s$}} (2);
            \path [olive, thick, dashed, ->] (3) edge[out=15,in=165] node[above] {{$d_t$}} (6);
        \end{tikzpicture} 
        \hspace*{1em}
        \begin{tikzpicture}[node distance={15mm}, thick, main/.style = {draw, circle}, scale=0.3] 
            \node[main, minimum size=0.5cm] (2) {$_{v_s}$ }; 
            \node[main, minimum size=0.5cm] (3) [above right of = 2] {$_{v_i}$ }; 
            \node[main, minimum size=0.5cm] (4) [above left of = 3] {$_{v_\ell}$ };
            \node[main, fill=gray!50, minimum size=0.5cm] (6) [right of = 3, right=-2em] {$_{s}$ };
            \node[main, minimum size=0.5cm] (5) [right of = 6, right=-2em] {$_{v_t}$ }; 
            \draw (2) [bend left, dashed] to node[left] {$e_{mid}$} (4);
            \draw (2) [bend right] to node[right] {} (3);
            \draw (3) [bend right] to node[above right] {} (4);
            \draw (3) [] to node[above] {} (6);
            \draw (6) [] to node[above] {} (5);
            \path [red, thick, dashed, ->] (3) edge[out=90,in=0] node[above] {{$d_\ell$}} (4);
            \path [blue, thick, dashed, ->] (3) edge[out=270,in=0] node[below] {{$d_s$}} (2);
            \path [violet, thick, dashed, ->] (6) edge[out=150,in=30] node[above] {{$d_i$}} (3);
            \path [olive, thick, dashed, ->] (6) edge[out=25,in=155] node[above] {{$d_t$}} (5);
        \end{tikzpicture} 
        \hspace*{1em}
        \begin{tikzpicture}[node distance={15mm}, thick, main/.style = {draw, circle}, scale=0.3] 
            \node[main, minimum size=0.5cm] (2) {$_{v_i}$}; 
            \node[main, fill=gray!50, minimum size=0.5cm] (5) [below left of = 4] {$_{s}$ }; 
            \node[main, minimum size=0.5cm] (3) [above right of = 2] {$_{v_\ell}$ }; 
            \node[main, minimum size=0.5cm] (4) [above left of = 3] {$_{v_s}$ };
            \node[main, minimum size=0.5cm] (6) [below of = 2] {$_{v_t}$ };
            \draw (2) [bend left] to node[left] {} (5);
            \draw (5) [bend left] to node[left] {} (4);
            \draw (2) [bend right] to node[right] {} (3);
            \draw (3) [bend right, dashed] to node[right] {$e_{mid}$} (4);
            \draw (2) [] to node[above] {} (6);
            \path [blue, thick, dashed, ->] (5) edge[out=90,in=180] node[left] {{$d_s$}} (4);
            \path [red, thick, dashed, ->] (5) edge[looseness=1.2,out=320,in=220] node[above] {{$d_\ell$}} (3);
            \path [olive, thick, dashed, ->] (2) edge[out=315,in=45] node[right] {{$d_t$}} (6);
            \path [violet, thick, dashed, ->] (5) edge[out=270,in=180] node[left] {{$d_i$}} (2);
        \end{tikzpicture} 
        }
        \end{center}
        \caption[Case1]{The possible starting positions in a tadpole graph.}
        \label{Tail}
    \end{figure}

\paragraph*{Starting at the end of the tail}

When starting at the end of the tail, both agents can be sent to the intersection, explore the cycle using an optimal $1.5$-competitive strategy
and return to $s$. This also leads to a $1.5$-competitive strategy. Since the paths taken to explore the cycle are identical to an optimal exploration,
this case is also $1$-competitive in the energy model.

\begin{align*}
    \opt &= 2d_t + \opt_C\\
    \onl &\leq 2d_t + 1.5\opt_C \leq 1.5\opt
\end{align*}

\paragraph*{Starting at the intersection}

When starting on the intersection, we randomly choose two outgoing edges and have each agent follow one path while applying the AMP strategy.
There are three possibilities for the paths chosen: $p_s$ and $p_\ell$, $p_s$ and $p_t$, or $p_\ell$ and $p_t$.\\
When $p_s$ and $p_\ell$ are chosen, which means exploring the cycle first, the agents explore the cycle and return. 
The first agent reaching $s$ again (which is the one that traversed $p_s$) will then explore the tail of the graph. 
This leads to a maximum exploration time of $d_s + d_\ell + \max(d_s+2d_t, d_\ell) \leq 5\max(d_\ell, d_t)$.\\

When $p_\ell$ and $p_t$ are chosen, one agent explores the tail, while another agent starts exploring the cycle. When the cycle is fully explored,
before the end of the tail is reached, the remaining agent then explores the rest of the tail and returns to $s$, leading to an exploration time less than
$3d_t \leq 5\max(d_\ell, d_t)$. If the end of the tail is reached before the cycle is fully explored, the agent exploring $p_t$ returns to $s$ and starts
exploring the remaining path $p_s$. When $d_t < d_\ell + l(e_{mid})$, the entire exploration time is at most $2d_t + 2d_\ell + d_s \leq 5\max(d_\ell, d_t)$.
If $d_t \geq d_\ell + l(e_{mid})$, the entire exploration time is at most $3d_t + 2d_s \leq 5\max(d_\ell, d_t)$. When choosing $p_s$ and $p_t$ we get 
exploration times of at most $3d_t$ (when $L \leq d_t$), $2d_t + 2d_\ell + d_s$ (when $d_t < d_s$), or $3d_t + 2d_\ell$ (when $d_t \geq d_s + l(e_{mid})$),
which are all less or equal to $5\max(d_\ell, d_t)$

\paragraph*{Starting on a node from the tail}

When starting on a node from the tail, one agent starts traversing in each direction, following the AMP strategy. When the intersection is found by
some agent, it randomly chooses an outgoing path. This leads to two possibilities: $p_t, p_\ell$ and $p_t, p_s$. In both cases, if $d_t \geq d_i + L$,
the exploration takes at most time $3d_l$. When $p_t$ and $p_\ell$ are chosen and $p_t \geq d_i + d_\ell + l(e_{mid})$, the exploration time is at most 
$3d_t + 2d_i + 2d_s$. If $p_t < d_i + d_\ell + l(e_{mid})$, the exploration takes at most time $2d_t + 3d_i + 2d_\ell + d_S$. When $p_t$ and $p_s$ are chosen
the exploration times are $3d_t + 2d_i + 2d_\ell$ (if $p_t \geq d_i + d_s + l(e_{mid})$) and $2d_t + 3d_i + 2d_\ell + d_S$ (if $p_t < d_i + d_s + l(e_{mid})$) respectively.

\paragraph*{Starting on a node from the cycle}

When starting on the cycle, there are two possibilities: $v_i$ can be located on either $p_s$ or $p_\ell$. In both cases the combination $p_\ell, p_s$ can be chosen,
which leads to maximum exploration times of $d_\ell + d_s + \max(d_s, d_\ell + 2d_t)$ (when $v_i$ is on $p_\ell$), or $d_\ell + d_s + \max(d_\ell, d_s + 2d_t)$ (when $v_i$ is on $p_s$).
We now cover the remaining cases: When $v_i$ is on $p_s$ and $p_\ell, p_t$ are chosen and $d_i + d_t \geq L - d_i$, meaning that the cycle is fully
explored before $v_t$ is reached, the exploration takes time at most $2(d_i + d_t)$. If $d_i + d_t < d_\ell + l(e_{mid})$, the exploration takes time $2d_t + 2d_\ell + d_s$.
If $d_i + d_t \geq d_\ell + l(e_{mid})$ the exploration takes time at most $3d_t + 2d_s + d_i$. When $v_i$ is on $p_\ell$ the exploration times are 
$2(d_i + d_t)$ (if $d_i + d_t \geq L - d_i$), $2d_t + 2d_\ell + d_s$ (if $d_i + d_t < d_s + l(e_{mid})$) and $3d_t + 2d_\ell + d_i$ (if $d_i + d_t \geq d_s + l(e_{mid})$) respectively.

\begin{figure}[H]
    \begin{center}
        \begin{tikzpicture}[node distance={15mm}, thick, main/.style = {draw, circle}] 
            \node[main, fill=gray!50, minimum size=0.5cm] (3) [] {$_{s}$ }; 
            \node[main, minimum size=0.5cm] (4) [left of = 3] {$_{v_m}$ };
            \node[main, minimum size=0.5cm] (6) [right of = 3] {$_{v_t}$ };
            \draw (3) [color=blue, dashed, bend right, out=270, in=270] to node[above] {$d_s = 1$} (4);
            \draw (3) [color=red, dashed, bend left, out=90, in=90] to node[below] {$d_\ell = 1$} (4);
            \draw (3) [color=purple, dashed] to node[above] {$d_t = 1$} (6);
        \end{tikzpicture} 
    \end{center}
    \caption[Tadpole example]{An example of a tadpole graph, where the competitive ratio can reach $2.5$. The dashed lines are paths with edges of length $\varepsilon$.
     If two agents randomly choose $d_s, d_\ell$ at the start of the exploration, the online exploration time reaches $5$, while the optimal exploration time with two agents is $2$.}
    \label{Example_randomchoice}
\end{figure}
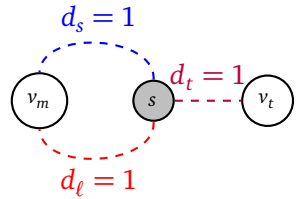

\paragraph*{Example in which a competitive ratio of exactly $2.5$ is reached}

Figure \ref{Example_randomchoice} depicts a graph in which the described exploration strategy reaches a competitive ratio of exactly $2.5$, if $p_s$ and $p_\ell$ are chosen at the
start of the exploration. In this case the cycle is explored, with the agents meeting at $v_{mid}$, before the agents backtrack to $s$ and then explore the tail of the graph. The exploration time 
is $5$. An optimal strategy would have sent one agent to $v_t$ and back, while the other agent explores the cycle, leading to an exploration time of $2$.
\end{proof}

\subsection{Tadpole Graphs with three Agents}\label{tadpole_three}
We modify the \textit{AMP} algorithm from \S\ref{cycles}, to construct a strategy for exploring tadpole graphs with $k=3$ agents. 
In the following, we first describe this strategy and then show that it achieves an optimal competitive ratio of $1$ for the energy model
as well as a competitive ratio of $2$ for the time~model.

\begin{theorem}[Three agent tadpole graph exploration]\label{TadpoleUpperThree}
    For $k=3$ agents there exists an exploration strategy for tadpole graphs, which has a competitive ratio of $2$ for the time model, and $1$ for the energy model.
\end{theorem}

\begin{proof}
Similar to the proof of Thm.\ref{TadpoleUpperTwo} we show that the competitive ratios stated above hold for each possible starting position
in a tadpole graph when using three agents. Since a strategy with two agents already optimally explores the graph when starting on $v_t$,
we can focus on the remaining cases: $s = v_i$ and $s$ is a node with degree two.

\paragraph*{Starting on the intersection} 
When the starting node $s$ is the intersection, an exploration strategy can send one agent in each direction. Like in the \textit{AMP} strategy for cycles, only one
agent traverses an edge at a time, minimizing the maximum distance each agent has traversed from $s$ to its current position. Because of this, the cycle is explored
optimally for the energy model, with the agents following exactly the same paths as in the \textit{AMP} algorithm for cycles. The tail is also optimally explored for
the energy model, since the agent exploring the tail can only follow the same path as an agent would in an offline strategy, leading to a competitive ratio of $1$ for the energy model.
For the time model, we reach a competitive ratio of $2$, because the strategy takes time $d_t + d_\ell + d_s + \max(d_t,d_\ell) \leq 4\max(d_t,d_\ell)$, since only one agent moves
at a time until the graph is explored, before all agents return to $s$ in parallel.

\paragraph*{Starting on a node with degree two}
When starting on a node with degree two, one agent waits at the node $s$ until the intersection is found, while the other two agents use the \textit{AMP} strategy to
explore the graph in both directions. As soon as the intersection $v_i$ is discovered, the third agent moves to the intersection, while the other agents wait. When the third agent has reached the intersection,
all three agents explore the graph in different directions, following the \textit{AMP} strategy. Since the agent discovering the intersection takes the shortest path with length $d_i$, 
and the agents proceed as in case 2, the strategy still reaches a competitive ratio of $1$ for the energy model, as well as $2$ for the time model. When the agents start on the cycle, they
traverse a distance of $d_s + d_\ell + (d_i + d_t) + \max(d_i + d_t, d_\ell) \leq 4\max(d_i + d_t, d_\ell)$. When the agents start on the tail, they traverse a distance of 
$(d_i + d_s) + (d_i + d_\ell) + d_t + \max(d_t, d_i + d_\ell) \leq 4\max(d_t, d_i + d_\ell)$.
\end{proof}

\subsection{Tadpole Graphs with four Agents}\label{tadpole_four}

When using four agents to explore a tadpole graph, we can achieve a $1.5$-competitive exploration strategy for the time model,
as well as a $1$-competitive strategy for the energy model. Since the optimal offline strategy cannot yield a better result than $2\max(d(s,v)), v\in V$,
we can still use the strategy given in \S\ref{tadpole_three} to reach $1$-competitiveness for the energy model. Also, since the lower bound for the exploration of 
tadpole graphs is $1.5$ as shown in \S\ref{lower_bound_tad}, using more than four agents for the exploration cannot yield a better competitive ratio than this strategy.

\begin{theorem}[Exploring tadpole graphs with four agents]\label{TadpoleUpperFour}
    Using $k=4$ agents, a competitive ratio of $1.5$ can be achieved on tadpole graphs in the time model.
\end{theorem}

\begin{proof}
    Since the optimal offline strategy needs $2\max(d(s,v), v\in V)$ time to explore a tadpole graph, and our strategy needs time $\max(d(s,v), v\in V)$ to backtrack after
    the last node has been visited, we can show that our online strategy needs at most time $2\max(d(s,v), v\in V)$ until all nodes of the graph have been visited, to achieve a competitive
    ratio of $1.5$ for the entire exploration.\\
    When the agents start the exploration on the intersection, three agents are sent into the graph, one in each possible direction. If they start on a node with degree two, pairs of 
    two agents (behaving like a single agent, making all steps in parallel) are sent to both directions. As soon as the intersection is found by a pair of agents, it splits, so that at 
    least one agent is exploring each path.\\
    The agents start the exploration by comparing the entire traversed distance after adding the next edge, like in the AMP strategy, but in this case all agents but the
    one with the longest total traversed distance are moving. When two agents would have traversed the same distance, both agents move. Whenever some agent reaches a node, if another agent is currently waiting, they compare the distances after the next edge
    traversals and the one seeing a shorter distance immediately starts traversing the next edge, so that at most one agent is waiting at any point in time.\\
    Since the agents on all other currently visible paths are moving forward whenever some agent $a_k$ waits, the maximum time needed until it traverses a distance of $d_k$ is $(d_k + t) \leq 2d_k$,
    where $t$ is the time waited. This is the case because as soon as $a_k$ has waited for a time $t=d_k$, the other agents must have either traversed a distance of $d_k$, 
    or already finished their exploration, which leads to $a_k$ starting the traversal until it reaches the same distance without stopping.\\
    Since the agents are exploring the graph on all possible paths, each node $v\in V$ is reached by some agent after at most time $2d(s, v)$ has passed. Because of this the last visited node 
    is reached by some agent after at most time $2\max(d(s,v), v\in V)$ has passed.
\end{proof}

\section{Discussion} \label{discussion}

We have shown that two agents suffice to optimally explore cycles for the energy model by modifying the \textit{ALE} strategy by Higashikawa et al.\ \cite{Higashikawa2012}
to create the \textit{AMP} strategy, which we also used to achieve optimal online exploration strategies on tadpole graphs with 3 agents for the energy and 4 agents for the time model.
An interesting result for the exploration of tadpole graphs with multiple agents is, that while on cycles online multi-agent strategies with an optimal competitive ratio of $1.5$ 
only achieve a worse competitive ratio than an optimal single-agent exploration strategy,
which reaches a competitive ratio of $\frac{1 + \sqrt{3}}{2} \simeq 1.366$ \cite{MIYAZAKI2009}, better competitive ratios
than in the single-agent case can be achieved on tadpole graphs, with an optimal competitive ratio of $1.5$ with four agents, versus $2$ with a single agent \cite{Brandt2020}. 
In the Appendix we sketch how the exploration strategies on tadpole graphs can be extended onto a graph class we call $n$-tadpole graphs, which are generalized variants of tadpole graphs, where $n$ tails are attached to a cycle instead of only one, with cycles being considered as $0$-tadpole graphs.

\paragraph*{Outlook}
We investigated cycles and tadpole graphs, and sketched extensions of those strategies for $n$-tadpole graphs
in the Appendix. 
These graph classes are all subclasses of unicyclic graphs.
Thus, one possibility for further research could be analyzing how exploration strategies for multi-agent exploration of unicyclic graphs can be achieved. But there are still other graph classes, which could also be considered for multi-agent graph exploration,
like directed graphs or cactus graphs.\\
Moreover, as seen in the results, there is still a gap between the upper and lower bounds in some scenarios.
Hence, another possibility for further research could be proving better lower bounds for tadpole graphs, or finding strategies matching the lower bound with only two or three~agents.\\
Lastly, we considered global communication between all agents, which allowed the identification
of the cycle edge containing the midpoint of the cycle. The strategies presented in this work do
not work if agents cannot communicate their already traversed distance. Thus, an open question is, how the competitive
ratio for the exploration of cycles and tadpole graphs is affected when restricting the communication between the agents.
\newpage
\bibliographystyle{ieeetr} 
\bibliography{literature} 
\newpage
\appendix
\section{Appendix}\label{appendix}
We now provide the deferred proof of Theorem~\ref{AleUpper} in the following Section \S\ref{proofAleUpper} and the extension of our strategies to $n$-tadpole graphs in the subsequent Section \S\ref{ntadpoles}.
For this proof, we consider $e_{max} = (v_s, v_\ell)$, as shown in Figure \ref{Example2}.
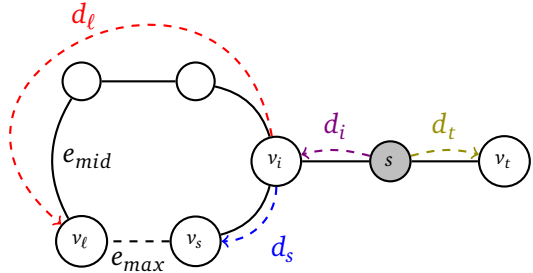
\begin{figure}[H]
    \begin{center}
        \begin{tikzpicture}[node distance={15mm}, thick, main/.style = {draw, circle}, scale=0.5] 
            \node[main, minimum size=0.5cm] (2) {$_{v_s}$ }; 
            \node[main, minimum size=0.5cm] (3) [above right of = 2] {$_{v_i}$ }; 
            \node[main, minimum size=0.5cm] (4) [above left of = 3] {};
            \node[main, fill=gray!50, minimum size=0.5cm] (6) [right of = 3] {$_{s}$ };
            \node[main, minimum size=0.5cm] (7) [left of = 2] {$_{v_\ell}$ };
            \node[main, minimum size=0.5cm] (8) [left of = 4] {};
            \node[main, minimum size=0.5cm] (9) [right of = 6] {$_{v_t}$ };
            \draw (7) [bend left] to node[right] {$e_{mid}$} (8);
            \draw (2) [bend right] to node[right] {} (3);
            \draw (3) [bend right] to node[above right] {} (4);
            \draw (3) [] to node[above] {} (6);
            \draw (2) [dashed] to node[below] {$e_{max}$} (7);
            \draw (4) [] to node[above] {} (8);
            \path [red, thick, dashed, ->] (3) edge[looseness=2.6, bend angle=90, out=100,in=140] node[above] {{$d_\ell$}} (7);
            \path [blue, thick, dashed, ->] (3) edge[out=270,in=0] node[below right] {{$d_s$}} (2);
            \path [olive, thick, dashed, ->] (6) edge[out=15,in=165] node[above] {{$d_t$}} (9);
            \path [violet, thick, dashed, ->] (6) edge[out=165,in=15] node[above] {{$d_i$}} (3);
            \draw (6) [] to node[above] {} (9);
        \end{tikzpicture} 
    \end{center}
    \caption[Tadpole example]{Tadpole Graph Example. The paths $p_\ell, p_s, p_i, p_t$ and the edges $e_{mid}$ and $e_{max}$ are marked.}
    \label{Example2}
\end{figure}

\subsection{Proof of Thm.~\ref{AleUpper}}\label{proofAleUpper}
\begin{proof} Consider the possible starting positions of the agents in a tadpole graph. We show that for any starting position a team of three agents will
    reach a competitive ratio of at most $3$, while a team of four agents reaches a competitive ratio of $2$.

\paragraph*{Starting at the end of the tail}
In this case two agents suffice to reach a competitive ratio of $1.5$. We consider $\opt_C$ to be the time needed for the exploration of the cycle by an optimal strategy.
We already know that the \textit{ALE} strategy needs at most time $1.5\opt_C$ for the exploration of a cycle. Since the exploration starts on $v_t$, an online strategy can send both agents
to $v_i$, let them explore the cycle using the \textit{ALE} strategy, and have them return to $v_t$. This leads to a competitive ratio of at most $1.5$.

\begin{align*}
    \opt &= 2d_t + \opt_C\\
    \onl &\leq 2d_t + 1.5\opt_C \leq 1.5\opt
\end{align*}

\begin{figure}[H]
    \begin{center}
        \begin{tikzpicture}[node distance={15mm}, thick, main/.style = {draw, circle}, scale=0.5] 
            \node[main, minimum size=0.5cm] (2) {$_{v_s}$ }; 
            \node[main, fill=gray!50, minimum size=0.5cm] (3) [above right of = 2] {$_{s}$ }; 
            \node[main, minimum size=0.5cm] (4) [above left of = 3] {$_{v_\ell}$ };
            \node[main, minimum size=0.5cm] (6) [right of = 3] {$_{v_t}$ };
            \draw (2) [bend left, dashed] to node[left] {$e_{max} = e_{mid}$} (4);
            \draw (2) [bend right] to node[right] {} (3);
            \draw (3) [bend right] to node[above right] {} (4);
            \draw (3) [] to node[above] {} (6);
            \path [red, thick, dashed, ->] (3) edge[out=90,in=0] node[above right] {{$d_\ell$}} (4);
            \path [blue, thick, dashed, ->] (3) edge[out=270,in=0] node[below right] {{$d_s$}} (2);
            \path [olive, thick, dashed, ->] (3) edge[out=15,in=165] node[above] {{$d_t$}} (6);
        \end{tikzpicture} 
        \hspace*{2em}
        \begin{tikzpicture}[node distance={15mm}, thick, main/.style = {draw, circle}, scale=0.5] 
            \node[main, minimum size=0.5cm] (2) {$_{v_s}$ }; 
            \node[main, fill=gray!50, minimum size=0.5cm] (3) [above right of = 2] {$_{s}$ }; 
            \node[main, minimum size=0.5cm] (4) [above left of = 3] {};
            \node[main, minimum size=0.5cm] (6) [right of = 3] {$_{v_t}$ };
            \node[main, minimum size=0.5cm] (7) [left of = 2] {$_{v_\ell}$ };
            \node[main, minimum size=0.5cm] (8) [left of = 4] {};
            \draw (7) [bend left] to node[right] {$e_{mid}$} (8);
            \draw (2) [bend right] to node[right] {} (3);
            \draw (3) [bend right] to node[above right] {} (4);
            \draw (3) [] to node[above] {} (6);
            \draw (2) [dashed] to node[below] {$e_{max}$} (7);
            \draw (4) [] to node[above] {} (8);
            \path [red, thick, dashed, ->] (3) edge[looseness=2.6, bend angle=90, out=100,in=140] node[above] {{$d_\ell$}} (7);
            \path [blue, thick, dashed, ->] (3) edge[out=270,in=0] node[below right] {{$d_s$}} (2);
            \path [olive, thick, dashed, ->] (3) edge[out=15,in=165] node[above] {{$d_t$}} (6);
        \end{tikzpicture} 
        \caption[ALE at intersection]{The two possible cases when the agents start at the intersection of the tadpole graph.\linebreak Left: $e_{mid}=e_{max}$, Right: $e_{mid}\neq e_{max}$.}
    \end{center}
    \label{Case2ALE}
\end{figure}
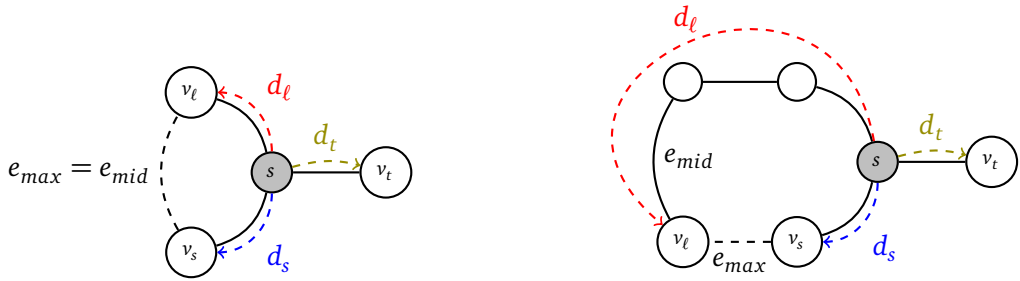

\paragraph*{Starting at the intersection}
In this case we can use three agents to reach a competitive ratio of $2$. This is achieved by sending one agent in each direction, where only the agent seeing the shortest
edge traverses this edge at any given time. If all three agents see edges of the same length the agent that will traverse the next edge is picked at random. If two agents see
the same edge all nodes on the cycle have been visited, the third agent must be located on the tail and completes the exploration of the tail. Since only one agent moves at any
given time, the agents exploring the cycle will always stop at both sides of one of the edges with the highest length in the cycle, which will then be considered to be $e_{max}$.\\
We denote the maximum distance traversed by an agent exploring the cycle in the online strategy until all nodes have been visited as $opt_\ell$. As shown in Figure \ref{Case2ALE}, we consider two cases: 
$e_{max} = e_{mid}$ and $e_{max} \neq e_{mid}$. In the first case, we know that $d_\ell = \opt_\ell \leq L/2$. This leads to a competitive ratio of at most~$2$. 

\begin{align*}
    \small
    \opt &= 2\max(d_t, d_\ell)\\
    \onl &= d_t + d_s + d_\ell + \max(d_t, d_\ell)\\
    &\leq 4\max(d_t, d_\ell) = 2\opt\\
\end{align*}

In the second case, where $e_{max} \neq e_{mid}$, we know that $d_s + l(e_{max}) \leq opt_\ell \leq L/2 \leq d_\ell$, since $p_\ell$ must include $e_{mid}$ to be the longer path to a neighbor of $e_{max}$.
Because of this, in the online case, when all nodes have been visited, the shortest distance to $s$ is $d_s + l(e_{max})$. We also know that $d_l + d_s = L - l(e_{max}) \leq L - l(e_{mid}) \leq 2\opt_\ell$.
Because of this, we also achieve a competitive ratio of $2$ in this case.

\begin{align*}
    \small
    \opt &= 2\max(d_t, \opt_\ell)\\
    \onl &= d_t + d_s + d_\ell + \max(d_t, d_s + l(e_{max}))\\ 
    \text{Case 1: } d_s + l(e_{max}) \leq \opt_\ell \leq d_t&:\\
    \opt &= 2d_t\\
    \onl &= d_s+d_\ell+2d_t \leq 2\opt_\ell + 2d_t\\ 
    &\leq 4d_t = 2\opt\\
    \text{Case 2: } d_s + l(e_{max}) \leq d_t \leq \opt_\ell&:\\
    \opt &= 2\opt_\ell\\
    \onl &= d_s + d_\ell + 2d_t \leq 2\opt_\ell + 2d_t\\
    &\leq 4\opt_\ell = 2\opt\\
    \text{Case 3: } d_t \leq d_s + l(e_{max}) \leq \opt_\ell&:\\
    \opt &= 2\opt_\ell\\
    \onl &= 2d_s + d_\ell + l(e_{max}) + d_t \leq d_s + d_\ell + 2\opt_\ell\\
    &\leq 4\opt_\ell = 2\opt
\end{align*}

\begin{figure}[H]
    \begin{center}
        \begin{tikzpicture}[node distance={15mm}, thick, main/.style = {draw, circle}, scale=0.5] 
            \node[main, minimum size=0.5cm] (2) {$_{v_s}$ }; 
            \node[main, minimum size=0.5cm] (3) [above right of = 2] {$_{v_i}$ }; 
            \node[main, minimum size=0.5cm] (4) [above left of = 3] {$_{v_\ell}$ };
            \node[main, fill=gray!50, minimum size=0.5cm] (6) [right of = 3] {$_{s}$ };
            \node[main, minimum size=0.5cm] (7) [right of = 6] {$_{v_t}$ };
            \draw (2) [bend left, dashed] to node[left] {$e_{max} = e_{mid}$} (4);
            \draw (2) [bend right] to node[right] {} (3);
            \draw (3) [bend right] to node[above right] {} (4);
            \draw (3) [] to node[above] {} (6);
            \draw (6) [] to node[above] {} (7);
            \path [red, thick, dashed, ->] (3) edge[out=90,in=0] node[above right] {{$d_\ell$}} (4);
            \path [blue, thick, dashed, ->] (3) edge[out=270,in=0] node[below right] {{$d_s$}} (2);
            \path [olive, thick, dashed, ->] (6) edge[out=15,in=165] node[above] {{$d_t$}} (7);
            \path [violet, thick, dashed, ->] (6) edge[out=165,in=15] node[above] {{$d_i$}} (3);
        \end{tikzpicture} 
        \begin{tikzpicture}[node distance={15mm}, thick, main/.style = {draw, circle}, scale=0.5] 
            \node[main, minimum size=0.5cm] (2) {$_{v_s}$ }; 
            \node[main, minimum size=0.5cm] (3) [above right of = 2] {$_{v_i}$ }; 
            \node[main, minimum size=0.5cm] (4) [above left of = 3] {};
            \node[main, fill=gray!50, minimum size=0.5cm] (6) [right of = 3] {$_{s}$ };
            \node[main, minimum size=0.5cm] (7) [left of = 2] {$_{v_\ell}$ };
            \node[main, minimum size=0.5cm] (8) [left of = 4] {};
            \node[main, minimum size=0.5cm] (9) [right of = 6] {$_{v_t}$ };
            \draw (7) [bend left] to node[right] {$e_{mid}$} (8);
            \draw (2) [bend right] to node[right] {} (3);
            \draw (3) [bend right] to node[above right] {} (4);
            \draw (3) [] to node[above] {} (6);
            \draw (2) [dashed] to node[below] {$e_{max}$} (7);
            \draw (4) [] to node[above] {} (8);
            \path [red, thick, dashed, ->] (3) edge[looseness=2.6, bend angle=90, out=100,in=140] node[above] {{$d_\ell$}} (7);
            \path [blue, thick, dashed, ->] (3) edge[out=270,in=0] node[below right] {{$d_s$}} (2);
            \path [olive, thick, dashed, ->] (6) edge[out=15,in=165] node[above] {{$d_t$}} (9);
            \path [violet, thick, dashed, ->] (6) edge[out=165,in=15] node[above] {{$d_i$}} (3);
            \draw (6) [] to node[above] {} (9);
        \end{tikzpicture}
        \caption[ALE at intersection]{The two possible cases when the agents start on the tail of the tadpole graph.\linebreak Left: $e_{mid}=e_{max}$, Right: $e_{mid}\neq e_{max}$.} 
    \end{center}
    \label{Case3ALE}
\end{figure}
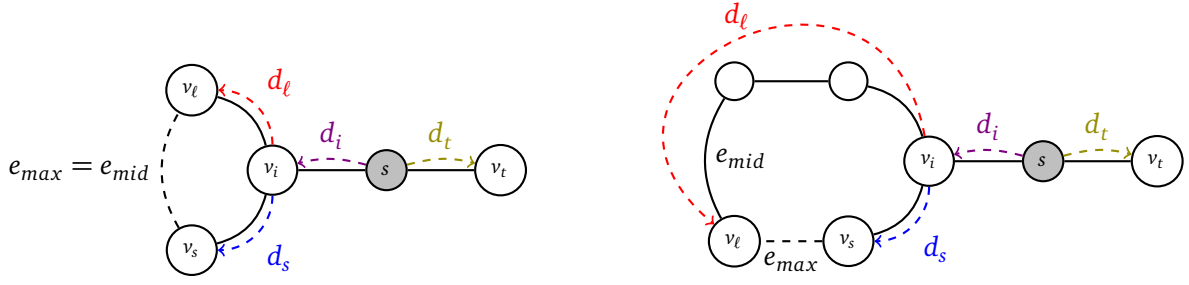

\paragraph*{Starting on a node from the tail}
In this case, which is shown in Figure \ref{Case3ALE}, we use three agents to reach a competitive ratio of $2$. At the beginning of the exploration,
one agent is sent in each of the two visible directions, following the \textit{ALE} strategy, while the third agents waits at $s$. As soon as $v_i$ is found, the
third agents traverses to the intersection while the other two agents wait. From that point on, the graph is explored in all visible directions until each node has been visited.
Like in the previous case, if $e_{mid} = e_{max}$, the paths traversed by both strategies match, and we achieve a competitive ratio of $2$.

\begin{align*}
    \small
    \opt &= 2\max(d_t, d_i + d_\ell)\\
    \onl &= d_t + (d_i + d_s) + (d_i + d_\ell) + \max(d_t, d_i + d_\ell)\\
    &\leq 4\max(d_t, d_i + d_\ell) = 2\opt\\
\end{align*}

If $e_{mid} \neq e_{max}$, the relationships between the paths in the cycle from the previous cases still hold, keeping the competitive ratio at $2$.

\begin{align*}
    \small
    \opt &= 2\max(d_t, d_i + \opt_\ell)\\
    \onl &= d_t + (d_i + d_s) + (d_i + d_\ell) + \max(d_t, d_i + d_s + l(e_{max}))\\
    \text{Case 1: } d_i + d_s + l(e_{max}) \leq d_i + \opt_\ell \leq d_t&:\\
    \opt &= 2d_t\\
    \onl &= 2d_t + 2d_i + d_\ell + d_s \leq 2d_t + 2\opt_\ell + 2d_i\\
    &\leq 4d_t = 2\opt\\
    \text{Case 2: } d_i + d_s + l(e_{max}) \leq d_t \leq d_i + \opt_\ell&:\\
    \opt &= 2\opt_\ell + 2d_i\\
    \onl &= 2d_t + 2d_i + d_\ell + d_s \leq 2\opt_\ell + 4d_i + d_\ell + d_s\\
    &\leq 4\opt_\ell + 4d_i = 2\opt\\
    \text{Case 3: } d_t \leq d_i + d_s + l(e_{max}) \leq d_i + \opt_\ell&:\\
    \opt &= 2\opt_\ell + 2d_i\\
    \onl &= d_t + 3d_i + 2d_s + d_\ell + l(e_{max})\\
    &\leq 4d_i + 2d_s + l(e_{max}) + d_\ell + \opt_\ell\\
    &= 4d_i + d_s + d_\ell + 2\opt_\ell\\
    &\leq 4d_i + 4\opt_\ell = 2\opt
\end{align*}

\begin{figure}[H]
    \begin{center}
        \begin{tikzpicture}[node distance={15mm}, thick, main/.style = {draw, circle}, scale=0.5] 
            \node[main, fill=gray!50, minimum size=0.5cm] (1) [] {$_{s}$ }; 
            \node[main, minimum size=0.5cm] (2) [above right of = 1] {$_{v_i}$ }; 
            \node[main, minimum size=0.5cm] (3) [below right of = 1] {$_{v_s}$ }; 
            \node[main, minimum size=0.5cm] (4) [below right of = 2] {$_{v_\ell}$ }; 
            \node[main, minimum size=0.5cm] (5) [above of = 2] {$_{v_t}$ }; 
            \draw (1) [bend left] to node[above] {} (2);
            \draw (1) [bend right] to node[above] {} (3);
            \draw (2) [bend left] to node[above] {} (4);
            \draw (3) [bend right, dashed] to node[right] {$e_{mid} = e_{max}$} (4);
            \draw (2) [] to node[above] {} (5);
            \path [red, thick, dashed, ->] (1) edge[looseness=1.25, out=45,in=135] node[below] {{$d_\ell$}} (4);
            \path [blue, thick, dashed, ->] (1) edge[out=270,in=180] node[left] {{$d_s$}} (3);
            \path [olive, thick, dashed, ->] (2) edge[out=100,in=260] node[left] {{$d_t$}} (5);
            \path [violet, thick, dashed, ->] (1) edge[out=90,in=180] node[left] {{$d_i$}} (2);
        \end{tikzpicture} 
        \hspace*{2em}
        \begin{tikzpicture}[node distance={15mm}, thick, main/.style = {draw, circle}, scale=0.5] 
            \node[main, fill=gray!50, minimum size=0.5cm] (1) [] {$_{s}$ }; 
            \node[main, minimum size=0.5cm] (2) [above right of = 1] {$_{v_i}$ }; 
            \node[main, minimum size=0.5cm] (3) [below right of = 1] {$_{v_s}$ }; 
            \node[main, minimum size=0.5cm] (4) [below right of = 2] {$_{v_\ell}$ }; 
            \node[main, minimum size=0.5cm] (5) [above of = 2] {$_{v_t}$ }; 
            \draw (1) [bend left] to node[above] {} (2);
            \draw (1) [bend right] to node[above] {} (3);
            \draw (2) [bend left] to node[right] {$e_{mid}$} (4);
            \draw (3) [bend right, dashed] to node[right] {$e_{max}$} (4);
            \draw (2) [] to node[above] {} (5); 
            \path [red, thick, dashed, ->] (1) edge[looseness=1.25, out=45,in=135] node[below] {{$d_\ell$}} (4);
            \path [blue, thick, dashed, ->] (1) edge[out=270,in=180] node[left] {{$d_s$}} (3);
            \path [olive, thick, dashed, ->] (2) edge[out=100,in=260] node[left] {{$d_t$}} (5);
            \path [violet, thick, dashed, ->] (1) edge[out=90,in=180] node[left] {{$d_i$}} (2);
        \end{tikzpicture} 
    \end{center}
    \caption[ALE at intersection]{The two possible cases when the agents start on the cycle of the tadpole graph.\linebreak Left: $e_{mid}=e_{max}$, Right: $e_{mid}\neq e_{max}$.}
    \label{Case4ALE}
\end{figure}
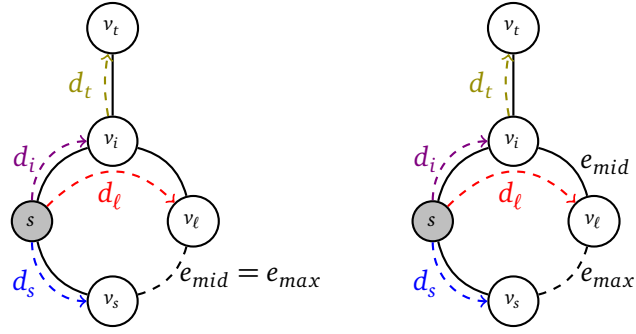

\paragraph*{Starting on a node from the cycle}
In this case, which is shown in Figure \ref{Case4ALE}, we use four agents to reach a competitive ratio of $2$, while using three agents leads to a competitive ratio of $3$. We first show the competitive ratio of $2$
for the four agent case. After that we show why a team of three agents only achieves a competitive ratio of $3$ in this case. With a team of four agents, two agents are sent in each direction, taking
all steps in parallel, until $v_i$ is found by one of the two agent teams. At this point, the team at $v_i$ splits and the exploration continues. Note that this case cannot be distinguished from the previous case
by an online algorithm because the agents only see a node with a degree of $2$ when starting the exploration, but since using this strategy saves exactly $d_i$, it is still $2$-competitive when
starting on a node of the tail. Like in the previous cases, if $e_{mid} = e_{max}$, the paths traversed by both strategies match, and we achieve a competitive ratio of $2$.
 
\begin{align*}
    \small
    \opt &= 2\max(d_t + d_i, d_\ell)\\
    \onl &= d_t + d_i + d_s + d_\ell + \max(d_t + d_i, d_\ell)\\
    &\leq 4\max(d_t + d_i, d_\ell) = 2\opt\\
\end{align*}

If $e_{mid} \neq e_{max}$, the relationships between the paths in the cycle from the previous cases still hold, keeping the competitive ratio at $2$.

\begin{align*}
    \small
    \opt &= 2\max(d_t + d_i, \opt_\ell)\\
    \onl &= d_t + d_i + d_s + d_\ell + \max(d_t + d_i, d_s + l(e_{max}))\\
    \text{Case 1: } d_s + l(e_{max}) \leq \opt_\ell \leq d_i + d_t&:\\
    \opt &= 2d_i + 2d_t\\
    \onl &= d_s + d_\ell + d_i + 2d_t \leq 2\opt_\ell + d_i + 2d_t\\
    &\leq 3d_i + 4d_t \leq 4d_i + 4d_t = 2\opt\\
    \text{Case 2: } d_s + l(e_{max}) \leq d_i + d_t \leq \opt_\ell&:\\
    \opt &= 2\opt_\ell\\
    \onl &= d_s + d_\ell + 2d_t + d_i \leq 2\opt_\ell + 2d_t + d_i\\
    &\leq 4\opt_\ell = 2\opt\\
    \text{Case 3: } d_i + d_t \leq d_s + l(e_{max}) \leq \opt_\ell&:\\
    \opt &= 2\opt_\ell\\
    \onl &= 2d_s + l(e_{max}) + d_\ell + d_t \leq (d_s + d_\ell) + d_t + \opt_\ell\\
    &\leq (L - l(e_{max})) + 2\opt_\ell \leq 4\opt_\ell = 2\opt
\end{align*}

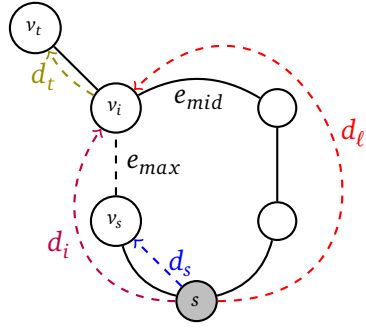
\begin{figure}[H]
    \begin{center}
        \begin{tikzpicture}[node distance={15mm}, thick, main/.style = {draw, circle}, scale=0.5] 
            \node[main, fill=gray!50, minimum size=0.5cm] (1) {$_{s}$ }; 
            \node[main, minimum size=0.5cm] (2) [above right of = 1] {}; 
            \node[main, minimum size=0.5cm] (3) [above left of = 1] {$_{v_s}$ }; 
            \node[main, minimum size=0.5cm] (4) [above of = 2] {}; 
            \node[main, minimum size=0.5cm] (5) [above of = 3] {$_{v_i}$ }; 
            \node[main, minimum size=0.5cm] (6) [above left of = 5] {$_{v_t}$ }; 
            \draw (1) [bend right] to node[above] {} (2);
            \draw (1) [bend left] to node[above] {} (3);
            \draw (2) [] to node[above] {} (4);
            \draw (3) [dashed] to node[right] {$e_{max}$} (5);
            \draw (5) [] to node[above] {} (6);
            \draw (4) [bend right] to node[below] {$e_{mid}$} (5);
            \path [red, thick, dashed, ->] (1) edge[looseness=2.7, out=0,in=45] node[right] {{$d_\ell$}} (5);
            \path [blue, thick, dashed, ->] (1) edge[out=135,in=315] node[right] {{$d_s$}} (3);
            \path [olive, thick, dashed, ->] (5) edge[out=150,in=300] node[left] {{$d_t$}} (6);
            \path [purple, thick, dashed, ->] (1) edge[looseness=1.3, out=180,in=235] node[left] {{$d_i$}} (5);
        \end{tikzpicture} 
    \end{center}
    \caption[ALE at intersection]{The special case, where $v_i$ is found on the longer path $L-d_i$ instead of the shortest path~$d_i$.}
    \label{Case5ALE}
\end{figure}

With only three agents available, the exploration strategy where the third agent waits at the start until the intersection is found only achieves a competitive ratio of $3$. Starting on a node of the cycle, there are two directions in which $v_i$ can be found. When $e_{mid} \neq e_{max}$, and $v_i$ is located on $p_\ell$,
with $L/2 < L - d_i \leq d_\ell$, the intersection can be found before the exploration of the cycle is completed, as depicted in Figure \ref{Case5ALE}. In this case the agent waiting at $s$ traverses
a distance of $L - d_i$ instead of the optimal distance of $d_i$. Since $v_i$ is located on $p_\ell$, $p_i$ must include $e_{max}$. Because of this
the exploration takes at most time $2\opt + (L - d_i) \leq 2\opt + (L - l(e_{max})) \leq 2\opt + 2\opt_\ell \leq 3\opt$, leading to a competitive upper bound of $3$ in the three agent case.
\end{proof}

\subsection{Extending the Strategies to $n$-Tadpole Graphs} \label{ntadpoles}
In this section we generalize the strategies found for tadpole graphs in \S\ref{tadpoles} onto $n$-tadpole graphs. In \S\ref{lower_bound_ntad} we show that the lower bound of $1.5$ also holds for the exploration
of $n$-tadpole graphs with any number of agents $k > 2$ in the time model. In \S\ref{ntad_two} we use the upper bound of $3$ shown by Fritsch \cite{Fritsch2021} for single-agent exploration of unicyclic graphs (graphs that contain exactly one cycle)
to construct a trivial upper bound of $12$ for two agent exploration of $n$-tadpole graphs with arbitrary $n$. This is possible, since every $n$-tadpole graph is also a unicyclic graph. In \S\ref{ntad_three} we extend the strategy using three agents on tadpole graphs to an $n+2$ agent strategy for $n$-tadpoles, reaching a competitive ratio of 
$1$ for the energy and $1.5+\frac{n}{2}$ for the time model. In \S\ref{ntad_four} the strategy using four agents on tadpole graphs is extended to a strategy using $2^{n+1}$ agents on $n$-tadpole graphs to achieve
a competitive ratio of $1.5$ for any $n$. Since any tadpole graph is also an $n$-tadpole graph, the lower bound of $1.5$ for a fixed number of two agents in the energy model
shown in Thm.\ref{TadpoleEnergyLower} still holds for $n$-tadpole graphs. The results of this section are listed in Table \ref{table:ntad}. We define $n$-tadpole graphs as follows:

\paragraph*{$n$-Tadpole Graph} An n-tadpole graph $T_n=(V,E)$ is a graph, which consists of a cycle to which
$n$ tails $T_1,...,T_n$ are attached. There can be multiple tails attached to the same cycle node. The length of the cycle will be denoted as $L_c$, the length of the tails
will be denoted as $L(T_i)$. Note that by this definition a cycle is a $0$-tadpole graph, and a tadpole graph is 
a $1$-tadpole graph. 

\begin{table}[H]
    \small
    \centering
    \begin{tabular}{|l|ll|ll|}
    \hline
               & \multicolumn{2}{c|}{Time Model}                                     & \multicolumn{2}{c|}{Energy Model}                                   \\ \hline
               & \multicolumn{1}{c|}{Lower Bound} & \multicolumn{1}{c|}{Upper Bound} & \multicolumn{1}{c|}{Lower Bound} & \multicolumn{1}{c|}{Upper Bound} \\ \hline
    $1$ Agent    & \multicolumn{1}{l|}{2 \cite{Brandt2020}}           & 3 \cite{Fritsch2021}                               & \multicolumn{1}{l|}{2 \cite{Brandt2020}}           & 3 \cite{Fritsch2021}                               \\ \hline
    $2$ Agents  & \multicolumn{1}{l|}{1.5 (Thm. \ref{nTadpoleLower})}         & 12 (Thm. \ref{nTadpoleUpperTwo})                               & \multicolumn{1}{l|}{1.5 (Thm. \ref{TadpoleEnergyLower})}            & 12 (Thm. \ref{nTadpoleUpperTwo})                               \\ \hline
    $n+2$ Agents & \multicolumn{1}{l|}{1.5 (Thm. \ref{nTadpoleLower})}         & $1.5+0.5n$ (Thm. \ref{nTadpoleUpperNplusTwo})                         & \multicolumn{1}{l|}{1}           & 1 (Thm. \ref{nTadpoleUpperNplusTwo})                                \\ \hline
    $2^{n+1}$ Agents   & \multicolumn{1}{l|}{1.5 (Thm. \ref{nTadpoleLower})}         & 1.5 (Thm. \ref{nTadpoleUpperExp})                            & \multicolumn{1}{l|}{1}           & 1 (Thm. \ref{nTadpoleUpperNplusTwo})                               \\ \hline
    \end{tabular}
    \caption{Results for multi-agent exploration of n-Tadpole Graphs.} \label{table:ntad}
\end{table}

\subsubsection{A lower bound for n-Tadpole Graphs}\label{lower_bound_ntad}

We can reuse the lower bound construction from \S\ref{lower_bound_tad} to show a lower bound of $1.5$ for the competitive ratio of any exploration of $n$-tadpole graphs,
with $k \geq 2$ Agents. We only need to show this for cases where $n\geq2$, since the lower bound of $1.5$ has already been shown for cycles by Higashikawa et al.\cite{Higashikawa2012}
and for tadpole graphs in \S\ref{lower_bound_tad}. For this, instead of one tail of length $\varepsilon$, we attach $n$ tails of length $\delta = \frac{\varepsilon}{n}$ to the starting node $s$.
Because of this the tails can still be explored in parallel to the exploration of the cycle, so the analysis from \S\ref{lower_bound_tad} still holds, since using more than
two agents for the exploration of a cycle cannot lower its exploration time.

\begin{theorem}[A lower bound for $n$-tadpole graphs]\label{nTadpoleLower}
    For any number of agents $k \geq 2$, there exists no online graph exploration algorithm with a lower competitive ratio than $1.5$ for the time model.
\end{theorem}
\vspace*{-1em}

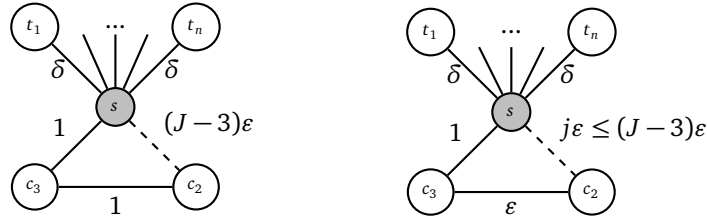
\begin{figure}[H]
    \small
    \begin{center}
    \begin{tikzpicture}[node distance={15mm}, thick, main/.style = {draw, circle}, scale=0.5] 
        \node[main, fill=gray!50, minimum size=0.5cm] (3) [] {$_{s}$ }; 
        \node[main, minimum size=0.5cm] (2) [below left of = 3] {$_{c_3}$ }; 
        \node[main, minimum size=0.5cm] (4) [below right of = 3] {$_{c_2}$ };
        \node[main, minimum size=0.5cm] (6) [above left of = 3] {$_{t_1}$ };
        \node[main, minimum size=0.5cm] (7) [above right of = 3] {$_{t_n}$ };
        \node (8) at ($(6)!0.5!(7)$) {...};
        \draw (2) [] to node[below] {$1$} (4);
        \draw (2) [] to node[above left] {$1$} (3);
        \draw (3) [dashed] to node[above right] {$(J-3)\varepsilon$} (4);
        \draw (3) [] to node[left] {$\delta$} (6);
        \draw (3) [] to node[right] {$\delta$} (7);
        \draw (3) [] -- ($(6)!0.3!(7) - (0,0.2)$);
        \draw (3) [] -- ($(6)!0.5!(7) - (0,0.2)$);
        \draw (3) [] -- ($(6)!0.7!(7) - (0,0.2)$);
    \end{tikzpicture} 
    \hspace*{5em}
    \begin{tikzpicture}[node distance={15mm}, thick, main/.style = {draw, circle}] 
        \node[main, fill=gray!50, minimum size=0.5cm] (3) [] {$_{s}$ }; 
        \node[main, minimum size=0.5cm] (2) [below left of = 3] {$_{c_3}$ }; 
        \node[main, minimum size=0.5cm] (4) [below right of = 3] {$_{c_2}$ };
        \node[main, minimum size=0.5cm] (6) [above left of = 3] {$_{t_1}$ };
        \node[main, minimum size=0.5cm] (7) [above right of = 3] {$_{t_n}$ };
        \node (8) at ($(6)!0.5!(7)$) {...};
        \draw (2) [] to node[below] {$\varepsilon$} (4);
        \draw (2) [] to node[above left] {$1$} (3);
        \draw (3) [dashed] to node[above right] {$j\varepsilon \leq (J-3)\varepsilon$} (4);
        \draw (3) [] to node[left] {$\delta$} (6);
        \draw (3) [] to node[right] {$\delta$} (7);
        \draw (3) [] -- ($(6)!0.3!(7) - (0,0.2)$);
        \draw (3) [] -- ($(6)!0.5!(7) - (0,0.2)$);
        \draw (3) [] -- ($(6)!0.7!(7) - (0,0.2)$);
    \end{tikzpicture} 
    \end{center}
    \small\caption[Tadpole lower bound]{Any number of agents $k\geq2$ cannot explore the given graph with a competitive ratio lower than $1.5$. The dashed line is a path consisting entirely of edges with length $\varepsilon$.
    Case 1 (Left): the graph an adversary reveals if $(s, c_3)$ is not traversed after an agent traversed a distance of $(J-3)\varepsilon$ on the path between $s$ and $c_2$.
    Case 2 (Right): the graph an adversary creates when $(s,c_3)$ is traversed before an agent traverses a distance of $j\varepsilon\leq(J-3)\varepsilon$ on the path between $s$ and $c_2$.}
    \label{ntadpole-time}
\end{figure}

\begin{proof}
    The analysis stays the same as in \S\ref{lower_bound_tad}, except that instead of a single tail of length $\varepsilon$, $n$ tails of combined length $\varepsilon$ are explored.
\end{proof}

\subsubsection{n-Tadpole graphs with two Agents}\label{ntad_two}
We can show that a constant competitive ratio regardless of $n$ can be achieved with only two agents on any $n$-tadpole graph. 
For this we make use of a competitive upper bound of $3$ for single-agent exploration of unicyclic graphs shown by Fritsch \cite{Fritsch2021},
to show that using only a single agent to explore an $n$-tadpole graph yields a competitive upper bound of $12$ compared to an optimal exploration using two agents.

\begin{theorem}[Upper bound for two-agent n-tadpole graph exploration]\label{nTadpoleUpperTwo}
    Using the single-agent strategy for the exploration of unicyclic graphs, leads to a competitive upper bound of $12$ compared to
    an optimal offline strategy using two agents. This applies to the time- as well as the energy model.
\end{theorem}

\begin{proof}
    Let $e_{max}$ denote the length of the longest cycle edge in a given $n$-tadpole graph:
    \begin{enumerate}
        \item If $l(e_{max}) > L_c/2$ an optimal single-agent offline strategy traverses a distance of $\opt_1 = 2(L_c-l(e_{max}))+2\sum_{i=1}^{n}L(T_i)$
        \item If $l(e_{max}) \leq L_c/2$ an optimal single-agent offline strategy traverses a distance of $\opt_1 = L_c+2\sum_{i=1}^{n}L(T_i)$
    \end{enumerate}

    To explore an $n$-tadpole graph, the exploration path must at least include all but one edges of the cycle, as well as all edges of all tails. So any optimal exploration
    has a length of at least $(L_c - l(e_{max}))+\sum_{i=1}^{n}L(t_i)$. If we assume that this can be evenly split between two agents, we get the shortest exploration time of 
    $\opt_2 = 0.5(L_c - l(e_{max}))+0.5\sum_{i=1}^{n}L(t_i)$ for two agent exploration of an $n$-tadpole graph. We now show that this is at most $4$ times smaller than $opt_1$.

    \paragraph*{Case 1} $2(L_c-l(e_{max}))+2\sum_{i=1}^{n}L(T_i) = 4\opt_2$.
    \paragraph*{Case 2} $L_c/2 \leq L_c - l(e_{max})$, since $l(e_{max}) \leq L_c/2$.\\
    \hspace*{1em} Thus, $L_c+2\sum_{i=1}^{n}L(T_i) \leq 2(L_c - l(e_{max})) + 2\sum_{i=1}^{n}L(T_i) = 4\opt_2$.\\

    Since the single-agent online exploration strategy for unicyclic graphs takes at most time $3\opt_1$, it takes at most time $12\opt_2$.
\end{proof}

\subsubsection{n-Tadpole graphs with $n+2$ agents}\label{ntad_three}

The strategy from \S\ref{tadpole_three}, which used three agents to explore a tadpole graph, can be generalized to a strategy using $n+2$ agents to explore any $n$-tadpole graph.
The general idea is, that instead of taking at most time $4\max(d(s,v), v\in V)$, resulting in a competitive ratio of $2$ for the time model, we take at most time $(n+3)\max(d(s,v), v\in V)$,
resulting in a competitive ratio of $1.5+\frac{n}{2}$ for the time model, while maintaining a competitive ratio of $1$ in the energy model, since the longest distance traversed by any agent is
still $2\max(d(s,v), v\in V)$, matching the offline strategy.

\begin{theorem}[Exploring $n$-tadpole graphs with $n+2$ agents]\label{nTadpoleUpperNplusTwo}
    A strategy for exploring an $n$-tadpole graph with $n+2$ agents can reach a competitive ratio of $1$ for the energy model, and
    a competitive ratio of $1.5+\frac{n}{2}$ for the time model.
\end{theorem}

\begin{proof}
    When starting the exploration, each outgoing edge of $s$ is assigned to an agent. The agents follow these edges one agent at a time, following the strategy from \S\ref{tadpole_three}.
    Whenever an intersection is found, one of the outgoing edges is assigned to the agent which found the intersection, while the other edges
    are assigned to agents without assigned edges. These agents will then traverse to the intersection while all other agents wait. After the agents reach
    the intersection, the exploration continues until all nodes have been visited by some agent and all agents return to $s$ by traversing the shortest paths.\\
    Since only one agent moves at any time, the paths traversed by the online strategy match the paths traversed in an optimal offline solution, for the same reasons
    as in \S\ref{tadpole_three}. Because of this, this strategy reaches a competitive ratio of $1$ for the energy model. While the optimal offline solution moves all agents in parallel, leading to an exploration time of $2\max(d(s,v), v\in V)$, the online strategy only
    moves the agents in parallel while backtracking, leading to a total exploration time of at most $(n+3)\max(d(s,v), v\in V)$ and a competitive ratio of $1.5+\frac{n}{2}$ 
    for the time model.
\end{proof}
\subsubsection{n-Tadpole graphs with $2^{n+1}$ agents}\label{ntad_four}

The $1.5$ competitive strategy for the exploration of tadpole graphs with four agents can be extended to a strategy using $2^{n+1}$ agents for the exploration
of $n$-tadpole graphs, which allows halving the teams of agents $n$ times.

    \begin{theorem}[Exploring $n$-tadpole graphs with $2^{n+1}$ agents]\label{nTadpoleUpperExp}
        A strategy for exploring an $n$-tadpole graph with $2^{n+1}$ agents can reach a competitive ratio of $1.5$ for the time model.
    \end{theorem}

\begin{proof} The strategy works similar to the one shown in \S\ref{tadpole_four}, but in this case, whenever an intersection with degree $j>2$ is found, the group of $k$ agents located on the intersection
    split into $j-1$ teams of size $\lfloor\frac{k}{j-1}\rfloor$ and continue the exploration in all unexplored outgoing directions. Since the initial team size
    is $2^{n+1}$, we never run out of agents when finding an intersection, while exploring an $n$-tadpole~graph.\\
We already know that any agent $a$ traversing some distance $d_a$, can wait at most $d_a$ time steps, before some other
agent reaches a node where the next traversed edge would increase its traversed distance to a value larger than $d_a$, or 
all other agents have reached their destinations. Because of this, the last visited node $v$ is visited after at most time $2d(s,v)$ has passed, 
leading to an exploration time of at most $3\max(d(s,v), v\in V)$ and a competitive ratio of $1.5$.
\end{proof}

\end{document}